\documentclass[lettersize, journal]{IEEEtran}
\usepackage{setspace}
\usepackage{cite}
\usepackage{pgfplots}
\pgfplotsset{compat=1.10}
\usepgfplotslibrary{fillbetween}
\usepackage{needspace}
\usepackage{balance}
\usepackage{bm}
\usepackage{lipsum}
\usepackage{makeidx}
\usepackage{enumerate}
\usepackage{color}
\usepackage{cite}
\usepackage{amsmath,amsthm}   
\usepackage{amssymb}
\usepackage{nomencl}
\usepackage{multirow}
\usepackage{graphicx}
\usepackage{epstopdf}
\usepackage{threeparttable}
\usepackage{multicol}
\usepackage{tikz}
\usetikzlibrary{matrix} 
\usetikzlibrary{positioning}
\DeclareGraphicsExtensions{.pdf,.jpeg,.png,.jpg,.emf,.eps}
\usepackage{algpseudocode}
\usepackage{algorithmicx}
\usepackage[ruled,noline,linesnumbered]{algorithm2e}
\usepackage{listings}
\usepackage[numbered,framed]{matlab-prettifier} 

\usepackage{calc}
\usepackage{here}
\usepackage{hyperref}

\usepackage{upref}
\usepackage{comment}
\usepackage{times}
\usepackage{dsfont}
\usepackage{epic,eepic}
\usepackage{rawfonts}
\usepackage[T1]{fontenc}
\usepackage{latexsym}
\usepackage{amsfonts}

\usepackage{capitalgreekitalic}

\usepackage{xcolor}
\usepackage{tikz,pgfplots}
\usetikzlibrary{calc}
\usetikzlibrary{intersections}
\usetikzlibrary{arrows,shapes}
\usetikzlibrary{positioning}

\newtheorem{theorem}{Theorem}

\newtheorem{proposition}{Proposition}

\newtheorem{remark}{Remark}

\def\Z{{\bf Z}}
\def\J{{\bf J}}
\def\R{{\bf R}}

\def\Q{{\bf Q}}
\def\I{{\bf I}}
\def\A{{\bf A}}
\def\E{{\bf E}}
\def\U{{\bf U}}
\def\V{{\bf V}}
\def\F{{\bf F}}
\def\T{{\bf T}}
\def\G{{\bf G}}
\def\H{{\bf H}}
\def\U{{\bf U}}
\def\B{{\bf B}}

\def\X{{\bf X}}
\def\Y{{\bf Y}}
\def\x{{\bf x}}

\def\f{{\bf f}}
\def\g{{\bf g}}

\def\Thetab{\bm{\Theta}}
\def\Sigmab{\bm{\Sigma}}
\def\Lambdab{\bm{\Lambda}}

\def\mub{\bm{\mu}}

\def\Gammab{\bm{\Gamma}}

\def\tr{\operatorname{tr}}
\def\det{\operatorname{det}}
\def\diag{\operatorname{diag}}
\def\Re{\operatorname{Re}}

\newcommand{\CU}{\mathcal{U}}

\allowdisplaybreaks

\begin{document}
\title{{The manifold of unitary and symmetric matrices: characterization, Riemannian optimization and application to BD-RIS design}
\author{Ignacio~Santamaria, \IEEEmembership{ Senior Member, IEEE}, Carlos~Beltr\'an, Eduard Jorswieck, \IEEEmembership{Fellow, IEEE}, Mohammad Soleymani, \IEEEmembership{Senior Member, IEEE}, Jes{\'u}s Guti{\'e}rrez, \IEEEmembership{Member, IEEE}}
\thanks{I. Santamaria is with the Department
of Communications Engineering, Universidad de Cantabria, 39005 Santander, Spain (e-mail: i.santamaria@unican.es).}
\thanks{C. Beltr\'an is with the Departamento de Matem\'aticas, Estad{\'i}stica y Computaci\'on, Universidad de Cantabria, 39005 Santander, Spain (e-mail: carlos.beltran@unican.es).}
\thanks{Eduard Jorswieck is with Institute for Communications Technology, Technische Universit{\"a}t Braunschweig, 38106 Braunschweig,
Germany (email: e.jorswieck@tu-bs.de).}
\thanks{M. Soleymani is with the Signal and System Theory Group, Universit{\"a}t  Paderborn, 33098 Paderborn, Germany (e-mail: mohammad.soleymani@uni-paderborn.de).}
\thanks{J. Guti{\'e}rrez is with IHP - Leibniz-Institut
f{\"u}r Innovative Mikroelektronik, 15236 Frankfurt (Oder), Germany (email: teran@ihp-microelectronics.com).}
\thanks{This work is partly supported by the European Commission’s Horizon Europe, Smart Networks and Services Joint Undertaking, research and innovation program under grant agreement 101139282, 6G-SENSES project. The work of I. Santamaria was also partly supported under grant PID2022-137099NB-C43 (MADDIE) funded by MICIU/AEI /10.13039/501100011033 and FEDER, UE. }
\thanks{A preliminary version of this work was presented at the 2026 IEEE International Conference on Acoustics, Speech, and Signal Processing [28].}
}
\maketitle

\begin{abstract}
This paper proposes and analyzes Riemannian optimization algorithms on the manifold of unitary and symmetric matrices, denoted $\CU_s$, which naturally models the scattering matrices of passive and reciprocal devices such as beyond-diagonal reconfigurable intelligent surfaces (BD-RISs). Despite its relevance, the geometry of $\CU_s$ has remained largely unexplored, and existing BD-RIS optimization methods either ignore the symmetry constraint or rely on costly Takagi-based parameterizations. We first provide a rigorous geometric characterization of $\CU_s$, deriving its tangent space, a simple retraction, and closed-form expressions for geodesics. Building on these results, we develop two Riemannian manifold optimization (MO) algorithms tailored to $\CU_s$: a line-search (LS) based scheme and a phase-optimization (PO) update along geodesics. We then apply the proposed framework to BD-RIS-assisted multiple-input multiple-output (MIMO) links, addressing sum-gain maximization, rate maximization, and minimum mean-square error problems, where they outperform existing approaches. Furthermore, we show that when the number of BD-RIS elements exceeds the total number of antennas, the optimal scattering matrix is low-rank, which motivates and enables efficient low-rank variants of the proposed algorithms.
\end{abstract}

\maketitle

\section{Introduction}
\label{sec:intro}
Riemannian optimization on matrix manifolds has become an essential tool in modern signal processing for handling non-Euclidean constraints that arise naturally in many estimation and design problems, such as orthogonality or subspace invariance \cite{AbsilBook},\cite{boumal2023intromanifolds},\cite{edelman1998geometry}, \cite[Ch. 9]{Coherence}. Key manifolds of interest include the Stiefel manifold, the Grassmann manifold, and the unitary manifold. Of particular interest in this work is the unitary manifold, denoted as $\CU$, which consists of unitary matrices that satisfy $\U\U^H=\I_n$. This manifold has been extensively studied and its geometry is well understood, allowing the development of effective Riemannian optimization algorithms with applications in communications \cite{SunTSP24}, blind source separation \cite{AbrudanTSP08},\cite{ABRUDAN20091704}, and machine learning \cite{ProjUNN}. 

However, in certain applications, the optimization matrix must be both unitary and symmetric ($\U = \U^T$), belonging to the manifold denoted in this paper as $\CU_s$. Symmetric unitary matrices have been studied in classical linear algebra and matrix completion problems \cite{MathisSIAM69}. In a classical work \cite{Dyson62}, F. Dyson showed that systems having time-reversal invariance and rotational symmetry are characterized by matrices in $\CU_s$. Symmetric unitary matrices also arise for modeling scattering matrices of reciprocal media in optics and microwaves \cite{Pozar11}. Our particular interest in unitary and symmetric matrices stems precisely from the fact that they model the scattering matrices of beyond-diagonal reconfigurable intelligent surfaces (BD-RISs). A BD-RIS is a 2D surface composed of numerous reconfigurable passive elements that can independently modify the phase of impinging waves \cite{ClerkxBDRISTut26}. It generalizes conventional diagonal-phase-shift RIS by interconnecting elements via tunable impedance networks, enabling non-diagonal scattering matrices for more flexible and powerful wave manipulation \cite{ClerckxTWC22b}, \cite{ClerckxTWC22a}. The BD-RIS scattering matrix, $\Thetab$, is a unitary (passive lossless) and symmetric (due to reciprocity) matrix \cite{ClerckxTWC22a, ClerckxTWC22b}, making it a natural candidate for optimization over $\CU_s$. Despite its relevance, $\CU_s$ has received little attention in the literature, and its main elements and properties such as the tangent space and geodesics remain unexplored. Previous approaches for BD-RIS optimization have circumvented the lack of characterization of $\CU_s$ by removing the symmetry constraint and optimizing over $\CU$, as in \cite{EmilEuCap2025, Zhao25, MarwaArxiv26}; or by relaxing the problem and projecting the relaxed solution onto the set $\CU_s$, as in \cite{MaoCL2024}. These alternatives are clearly suboptimal and provide only approximate solutions. Another approach, considered in \cite{SantamariaSPAWC24, Xia25, SantamariaOJVT25}, exploits the so-called Takagi factorization \cite{Takagi, Autonne}, which states that any unitary and symmetric matrix $\Thetab \in \CU_s$ can be decomposed as $\Thetab = \Q\Q^T$, where $\Q \in \CU$. However, applying Takagi’s decomposition and performing the optimization directly on $\CU$ unnecessarily increases the complexity of the optimization problem (for example, a quadratic cost function in $\Thetab$ becomes quartic in $\Q$, and, moreover, as we will see, the dimension of the search space $\CU$ is much higher than that of $\CU_s$), thereby increasing the number of possible local minima and hindering convergence. In \cite{fidanovski2026}, the BD-RIS cost function is augmented with a penalty term that enforces symmetry, but the optimization is still performed in $\CU$. Finally, we note that a unitary and diagonal $\Thetab$, which characterizes a diagonal RIS, is a special case of a matrix in $\CU_s$. However, the manifold of unitary and diagonal matrices is the well-known complex circle manifold, whose application to RIS design problems has been extensively studied \cite{Mossallamy21, Shtaiwi23}. For a recent comprehensive review of the complex circle manifold and its applications to the design of diagonal RIS, see \cite{Junior24}.

The main goal of this paper is to rigorously study the geometry of $\CU_s$ and leverage these results to design manifold optimization (MO) algorithms that operate directly on this manifold, as well as to present applications to BD-RIS optimization in multiple-input multiple-output (MIMO) links. A preliminary version of this paper has appeared in \cite{santamariaICASSP26}, where the tangent plane and the geodesics of $\CU_s$ are characterized for the first time. This work extends \cite{santamariaICASSP26}, presenting more complete and rigorous proofs of the mathematical results, developing two variants of Riemannian manifold optimization (MO) algorithms in $\CU_s$, revisiting the Cayley transform \cite{Postikov}, \cite{Hassibi02} and its restriction to $\CU_s$, and providing a more comprehensive study of applications to BD-RIS design problems. The main contributions of this work are:
\begin{itemize}
\item We present a study of the manifold $\CU_s$, characterizing its tangent space, a simple retraction mapping, and its geodesics. Additionally, we revisit the Cayley transform showing that, while it maps skew--Hermitian matrices to matrices in $\CU$, it also maps real and symmetric matrices, times the imaginary unit, to matrices in $\CU_s$, which defines a restriction of the Cayley transform to have codomain $\CU_s$. As a byproduct of this analysis, an alternative factorization to Takagi’s emerges for matrices in $\CU_s$.

\item We develop and compare two Riemannian optimization algorithms that exploit the parameterization of geodesics in $\CU_s$. The first one applies a standard gradient descent (or ascent) algorithm whose adaptation step is obtained using a line search (LS) procedure. The second, which we will name PO, does not require setting an adaptation step and optimizes the phases of the geodesic parameterization one by one.

\item As an application of the proposed Riemannian optimization framework, we consider three problems in BD-RIS-assisted MIMO systems: i) the maximization of the sum channel gain, ii) the maximization of the achievable rate, and iii) the minimization of the mean square error (MSE). For the three cost functions, we develop LS and PO MO algorithm versions and compare their performance with alternative approaches proposed in the literature.

\item When the number of BD-RIS elements is greater than the number of transmit and receive antennas in the MIMO link, that is, $M > N_t+N_r$, we show that the optimal BD-RIS for any optimization problem is precisely of rank at most $r =N_t+N_r$. This suggests the possibility, which is explored in the paper, of developing low-rank versions of the MO algorithms with significant computational savings.

\end{itemize}

The remainder of the paper is organized as follows. Sec.~\ref{sec:Cus} presents the main mathematical results characterizing the tangent space, the geodesics, and the retraction to the manifold $\CU_s$. We also study in Sec. \ref{sec:Cus} the restriction of the Cayley transform to unitary and symmetric matrices. By leveraging the mathematical characterization of the tangent space of $\CU_s$, Sec. \ref{sec:MOgeneric} develops a Riemannian optimization framework in $\CU_s$, which leads to two MO algorithms that we call Line Search (LS) and Phase Optimization (PO). Sec. \ref{sec:applic} applies the proposed MO algorithms to optimize BD-RIS-aided MIMO systems. It considers three common metrics in MIMO channels: the sum channel gain (squared Frobenius norm), achievable rate, and mean squared error (MSE). An analysis of the equivalent MIMO channel reveals a low-rank structure of the BD-RIS that can be exploited in the MO algorithms. Finally, Sec. \ref{sec:conclusion} concludes this paper, also outlining some future areas of work.

\indent \textit{Notation}: The symbols for scalars, vectors, matrices, and sets are, respectively, $x$, $\x$, $\X$, and $\mathcal{X}$. $\A^T$, $\A^*$, $\A^H$, $\A^{-1}$, $\A^{1/2}$, $\det(\A)$ are, respectively, transpose, conjugate, Hermitian, inverse, square root, and determinant of matrix $\A$. $\mathrm{Real}(a)$, $\mathrm{Imag}(a)$, and $\angle a$ are, respectively, the real part, imaginary part, and the angle of the complex number $a$, and $j$ denotes the imaginary unit. $\I_n$ denotes the identity matrix of size $n$. ${\cal CN}({\bf 0}, {\bf R})$ is the proper complex Gaussian distribution with zero mean and covariance matrix ${\bf R}$. The spectrum of a matrix $\A$ is the set of its eigenvalues and is denoted as $\mathrm{Spec}(\A)$, and its Frobenius norm is denoted as $\|\A \|_F$. $\CU=\{\U \in \mathbb{C}^{n\times n}: \U\U^H=\I_n\}$ denotes the unitary manifold. We also consider the manifold of unitary and symmetric matrices denoted as  $ \CU_s= \{\U \in \CU: \U=\U^T\}$. When it is necessary to specify the dimensions of the matrices, we will use $\CU(n)$ and $\CU_s(n)$.

\section{The manifold $\CU_s$}
\label{sec:Cus}
\subsection{Background: The manifold $\CU$} \label{sec:CU}
It is well known that $\CU=\{\U \in \mathbb{C}^{n\times n}: \U\U^H=\I_n\}$ is a real Riemannian manifold of dimension $n^2$ \cite{AbsilBook,boumal2023intromanifolds,edelman1998geometry}, with tangent space $T_{\U} \CU=\{\B\in \mathbb{C}^{n\times n}: \U^H\B+\B^H\U={\bf 0}\}$. Alternatively, the tangent space at $\U$ can be parametrized as $\B = \U {\bf S}$, where ${\bf S}$ is an $n \times n$ skew-Hermitian matrix.
The Riemannian exponential that maps a point in the tangent space to the manifold is defined in terms of the matrix exponential as
\[
\begin{matrix}
	\exp_{\CU,\U}:&T_{\U}\CU&\to&\CU\\
	&\B&\to&\U e^{\U^H\B},
\end{matrix}
\]
which is a diffeomorphism if restricted to tangent matrices $\B$ of sufficiently small Frobenius norm \cite{HallLieGroups}.

\subsection{Characterization of $\CU_s$: tangent space, geodesics, and retraction}
Our objective is to study the set $ \CU_s= \{\U \in \CU: \U=\U^T\}$ and then derive a Riemannian optimization algorithm. Firstly, we recall Takagi's decomposition of a complex and symmetric matrix, which plays an important role in the characterization of $\CU_s$. 
\begin{theorem}[Takagi factorization \cite{Takagi}]
Let $\A = \A^T$ be an $n \times n$ complex symmetric matrix. Then, there exist an $n \times n$ unitary matrix $\Q \in \CU$ and an $n\times n$ diagonal matrix $\bm{ \Sigma} = \diag(\sigma_1,\ldots,\sigma_n)$ with $\sigma_1 \geq \sigma_2 \geq \ldots \geq \sigma_n \geq 0$ such that $\A = \Q \bm{ \Sigma} \Q^T$.
\end{theorem}

From the singular value decomposition (SVD) of a symmetric matrix $\A = \F \Sigmab \G^H$, its Takagi decomposition can be obtained as $\A = \Q \Sigmab \Q^T$ with $\Q = \F (\F^H\G^*)^{1/2}$. Note that if all singular values of $\A$ are distinct,  $\F^H\G^*$ is a diagonal matrix and the factor $\Q$ is obtained as described in \cite[Remark 2]{SantamariaSPLetters2023}. We have the following propositions:
\begin{proposition} [{\bf Tangent space and geodesics}]
\label{prop:tangent}
	$\CU_s$ is a real Riemannian manifold of dimension $n(n+1)/2$ and its tangent space at the point $\U \in\CU_s$ is
	\begin{align*}
	T_{\U}\CU_s=&\{\B\in  \mathbb{C}^{n\times n}: \U^H\B+\B^H\U={\bf 0},\;\B=\B^T\}\\
    =& \{j\Q\R\Q^T:\R\in \mathbb{R}^{n\times n}, \R=\R^T \},
	\end{align*}
    where $\U=\Q\Q^T$ is a Takagi decomposition of $\U$. Therefore, the tangent space at $\U =  \Q\Q^T$ is composed of all matrices $\B = j \Q\R\Q^T$ with $\R$ an  $n\times n$ real and symmetric matrix. In addition, the Riemannian exponential is the restriction of $\exp_{\CU,\U}$ to $T_{\U}\CU_s$. Therefore, the geodesic in $\CU_s$ starting at $\U =\Q\Q^T$ and with direction $\B = j\Q\R\Q^T \in  T_{\U}\CU_s$ is
    \begin{equation*}
        \U(\mu) = \U e^{\U^H\B \mu} = \U e^{j \Q^* \R \Q^T \mu},
    \end{equation*}
    with $\R$ real and symmetric, and $\mu \geq 0$ a step size that parameterizes the geodesic.
\end{proposition}

\begin{proof}
    See Appendix.
\end{proof}

The above result can be rewritten as follows. The matrix $j\R$ (purely imaginary skew-Hermitian), admits the following eigendecomposition $j\R = \V_R \Sigmab_R \V_R^T$, where $\V_R$ is a (real) orthogonal matrix and $\Sigmab_R = \diag(j\theta_1,\ldots, j\theta_n)$ is a diagonal matrix with imaginary eigenvalues. Applying standard results of the exponentiation of diagonalizable matrices, we see that the geodesic starting at $\U = \Q\Q^T$ and pointing in the direction $j\Q\R\Q^T$ can be written as
\begin{eqnarray}
     \U(\mu) =  \U e^{j \Q^* \R \Q^T \mu} &=& \Q\Q^T \Q^* \V_R \Lambdab_R^\mu \V_R^T \Q^T  \nonumber \\
     &\stackrel{(a)}{=}& \Q \V_R \Lambdab_R^\mu \V_R^T \Q^T \nonumber \\
     & \stackrel{(b)}{=}& \Q_R  \Lambdab_R^\mu \Q_R^T, \label{eq:geodesic2}
\end{eqnarray}
where in $(a)$ we have used $\Q^T \Q^*  = \I_n$, in $(b)$ the unitary matrix $\Q \V_R$ is denoted as $\Q_R$, and $\Lambdab_R^\mu$ is the following diagonal matrix   
\begin{equation}
\Lambdab_R^\mu = \exp(\mu\Sigmab_R) = \diag \left(e^{j\theta_1 \mu},\ldots, e^{j\theta_n \mu} \right).
\label{eq:RISmatrix}
\end{equation}
In words, geodesics in $\CU_s$ are traversed by multiplying the phases of a diagonal matrix, $\Lambdab_R = \exp(\Sigmab_R)$, by a step size $\mu$. The direction of the geodesic, which depends on the cost function to be optimized, is determined by the eigenvectors and eigenvalues of $\R$. 

\begin{proposition} [{\bf Retraction to} $\CU_s$ ] \label{prop:retraction} The mapping 
\[
	\begin{matrix}
		\Pi:&\{ \A\in \mathbb{C}^{n\times n}: \A=\A^T\}&\to&\CU_s\\
		&\A&\to& \Q\Q^T,
	\end{matrix}
\]
where $\A=\Q \Sigmab \Q^T$ is a Takagi factorization of $\A$, sends $\A$ to the unitary and symmetric matrix $ \U = \Q\Q^T \in \CU_s$ closest to $\A$ in Frobenius norm (to one of them, if there are several).
\end{proposition}

\begin{proof}
    Given $\A = \F  \Sigmab  \G^H$, one can also write the polar decomposition $$\A=(\F \G^H)(\G \Sigmab \G^H).$$ It is a classical fact that the unitary factor in the polar decomposition, $\Pi(\A) = \F\G^H$, is the unitary matrix closest to $\A$, see for example \cite[Theorem 1]{FanHoffman}. If $\A=\Q \Sigmab \Q^T$ is symmetric, $\Pi(\A) = \Q\Q^T$ is hence the unitary and symmetric matrix closest to $\A$ in $\CU$ and consequently in $\CU_s \subset \CU$.
\end{proof}

\begin{proposition} [{\bf Projection to the tangent space}] \label{prop:projection} Given any matrix $\J \in \mathbb{C}^{n\times n}$, the orthogonal projection of $\J$ onto $T_{\U}\CU_s$ is 
\begin{equation}
		\pi_{T_{\U}\CU_s} \J = j\,\Q \underbrace{\mathrm{Imag}\left(\frac{\Q^{H}(\J+\J^T)\Q^{*}}{2}\right)}_{\R}\Q^T,
        \label{eq:proj}
\end{equation}	
	with $\U = \Q\Q^T$ a Takagi decomposition of $\U \in \CU_s$.
\end{proposition}

\begin{proof} According to Proposition \ref{prop:tangent} the tangent space at $\U = \Q\Q^T$ can be parameterized as $\B = j\Q\R\Q^T$ with $\R$ real and symmetric. Let us first solve the problem 
\begin{equation}\label{eq:JJT}
    \min_{\R \in \mathbb{R}^{n \times n}, \R = \R^T} \| {\bf J}- j\Q\R\Q^T \|_F^2, 
\end{equation}
or, equivalently,
\begin{equation*}
    \min_{\R \in \mathbb{R}^{n \times n}, \R = \R^T} \| -j\Q^H{\bf J}\Q^*-\R \|_F^2, 
\end{equation*}
that amounts to computing the real symmetric matrix $\R$ which is closest to $-j\Q^H{\bf J}\Q^*$. This is solved trivially by
\begin{align*}
\R=&\frac{\mathrm{Real}(-j\Q^H{\bf J}\Q^*)+\mathrm{Real}(-j\Q^H{\bf J}\Q^*)^T}{2}&\\
=&
\frac{\mathrm{Imag}(\Q^H{\bf J}\Q^*)+\mathrm{Imag}(\Q^H{\bf J}\Q^*)^T}{2}\\
=&\mathrm{Imag}\left(\frac{\Q^{H}(\J+\J^T)\Q^{*}}{2}\right).
\end{align*}
After inserting this value of $\R$ in the left hand side term of \eqref{eq:JJT}, we get
we get \eqref{eq:proj}.
\end{proof}

\subsection{Cayley's transform for matrices in $\CU_s$}
\label{sec:Cayley}
The Cayley transform \cite{Postikov} is a way to express a unitary matrix by means of a nonlinear transformation of an $n \times n$ skew-Hermitian matrix $j \B$, with $\B$ Hermitian\footnote{Note that $\B$ is Hermitian, if and only if $j\B$ is skew-Hermitian.}. It is given by
\begin{equation}
    \U = (\I_n + j\B)^{-1}(\I_n - j\B),
    \label{eq:B2Theta}
\end{equation}
where the matrix $(\I_n + j\B)$ is invertible because all the eigenvalues of the skew-Hermitian matrix $j\B$ have null real part, so they must be different from $-1$. Not all unitary matrices can be described this way: only those $\U$ with the property that all their eigenvalues are different from $-1$ admit a factorization of the form \eqref{eq:B2Theta}. This transform, which was discovered by Arthur Cayley in 1846 \cite{Cayley46}, has been extensively studied in several signal processing problems, ranging from image processing \cite{Kovacevic06} to codebook design in non-coherent communications \cite{Hassibi02}. In BD-RIS design problems, the scattering matrix and the admittance matrix are related by the Cayley transform \cite{ClerkxBDRISTut26},\cite{Nerini24}.

We are motivated by the following question: Is it possible to restrict the Cayley transform to reach only unitary and symmetric matrices? The following proposition answers the question, yielding a parametrization of (an open full measure subset of) unitary and symmetric matrices. Further, it also leads to a factorization of matrices in $\CU_s$ alternative to Takagi's method. 

Let us denote the set of $n \times n$ real symmetric matrices as ${\cal{B}} = \{\B:\B\in\mathbb R^{n\times n},\;\B=\B^T\}$, so that $j{\cal{B}}=T_{\I_n}\CU_s$, and the set of matrices in $\CU_s$ with all eigenvalues different from $-1$ as $\CU_{s}' = \{\U\in\CU_s:-1\not\in\mathrm{Spec}(\U)\}$.
\begin{proposition}[{\bf Restriction of the Cayley transform to matrices in $\CU_s$}] 
\label{prop:Cayley}  The restriction of the Cayley transform to the set of $j\B$ where $\B \in {\cal{B}} $ defines a diffeomorphism to $\CU_{s}'$: 
\begin{eqnarray*}
         \varphi:  {j\cal{B}} &\to&
         \CU_{s}' \\
    j\B&\mapsto& (\I_n + j\B)^{-1}(\I_n - j\B).  
\end{eqnarray*} 
The same formula gives its inverse,
\begin{eqnarray*}
         \varphi^{-1}:  \CU_{s}' &\to & j{\cal{B}}\\
    \U &\mapsto&  (\I_n + j\U)^{-1}(\I_n - j\U).
\end{eqnarray*} 
The mapping $\varphi$ hence sends the tangent space to $\CU_s$ at $\I_n$ to $\CU_s$. It differs from the matrix exponential: indeed, the derivatives at the origin are not equal. To show this, let $\varphi(tj\dot \B)$ be a smooth curve defined by the Cayley transform on the manifold ${\cal{B}}$, where we are denoting by $\dot\B$ any tangent vector to the space of real symmetric matrices (that is itself a real symmetric matrix)\footnote{You can think of $-2j\dot\B$ as a velocity (the time derivative in classical mechanics) for the curve $\varphi(tj\dot \B)$. At $t=0$, the point on the manifold is $\varphi(0) = \I_n$, and the derivative or velocity (on the tangent plane) is $\dot{\varphi}(0) = -2j\dot\B$. For example, take the curve $x(t)=e^{jt}$ on the complex circle $\lvert x \rvert=1$. The curve starts at $x(0)=1$ and the velocity is  $\dot{x}\mid_{t=0}=j e^{jt}\mid_{t=0}=j$ (the tangent vector).}. Note that $\varphi(tj\dot \B)\mid_{t=0} = \I_n$. Then, it is easy to show that
\[
\left. \frac{d}{dt}(\varphi(tj\dot \B)) \right |_{t=0}=-2j\dot\B.
\]
However, the Riemannian exponential for the manifold $\CU$ at $\U = \I_n$  is $\exp_{\CU,\I_n}(tj\dot \B) = \I_n \exp(t j \dot \B)$ (cf. Sec. \ref{sec:CU}). Therefore, its derivative at $t=0$ is 
\[
\left. \frac{d}{dt}(\exp_{\CU,\I_n}(tj\dot \B)) \right |_{t=0}=j\dot\B.
\]
\end{proposition}
\begin{proof}[Proof of Proposition \ref{prop:Cayley}]
    Let $\B=\V{\bf{\Sigma}}\V^T$ be a Schur decomposition of $\B$ where $\B$ is real and symmetric as in the hypotheses, and note that
    \begin{align*}
    [(\I_n + j\B)^{-1}(\I_n - j\B)]^T=&
    [\V(\I_n + j{\bf{\Sigma}})^{-1}(\I_n - j{\bf{\Sigma}})\V^T]^T\\=&
    \V(\I_n + j{\bf{\Sigma}})^{-1}(\I_n - j{\bf{\Sigma}})\V^T\\=&
    (\I_n + j\B)^{-1}(\I_n - j\B),
    \end{align*}
    where for the second equality we have used that diagonal matrices commute. Hence, the image of $j\B$ by the Cayley transform lies in $\CU_s$. Moreover, the inverse Cayley transform
    \[
    \U\mapsto(\I_n+\U)^{-1}(\I_n-\U)
    \]
    sends $\U\in\CU_s$ to $j\B$ with $\B$ real and symmetric, hence the map is a diffeomorphism.
\end{proof}


\begin{remark}[{\bf A decomposition for matrices in $\CU_s$ alternative to Takagi factorization}] \label{remark:decomp}
     In terms of the Schur decomposition $\B=\V{\bf{\Sigma}}\V^T$ of $\B$, where $\V$ and ${\bf{\Sigma}}$ are real matrices, $\V$ is orthogonal and ${\bf{\Sigma}}=\diag(\lambda_1,\ldots,\lambda_n)$, we can write
\begin{eqnarray*}
    \U =&\V(\I_n + j{\bf{\Sigma}})^{-1}(\I_n - j{\bf{\Sigma}})\V^T= \\
    &\V \diag\left(\frac{1-j\lambda_1}{1+j\lambda_1},\ldots,\frac{1-j\lambda_n}{1+j\lambda_n}\right)\V^T.
\end{eqnarray*}
    In particular, any matrix $ \U \in \CU_s$ can be decomposed in the form 
    \begin{equation}
        \U=\V{\bf \Lambda}\V^T, \label{eq:realdecomp}
    \end{equation}
    where $\V$ is a real orthogonal matrix and ${\bf \Lambda}=\diag(e^{j\theta_1},\ldots,e^{j\theta_n})$ is a diagonal matrix with entries in the unit circle $e^{j\theta_i} = \frac{1-j\lambda_i}{1+j\lambda_i}$. Note that this also covers the case that $\U$ has some eigenvalue equal to $-1$, which corresponds to the limit case that some eigenvalue(s) of $\B$ grow to $\infty$. 

    Eq. \eqref{eq:realdecomp} provides an alternative decomposition to Takagi factorization for unitary and symmetric matrices\footnote{The decomposition in \eqref{eq:realdecomp} has been used in \cite{NeriniTWC2023}, without a rigorous proof, to find a closed-form solution for the BD-RIS maximizing the received signal power in a single-input single-output (SISO) channel. Proposition \ref{prop:Cayley} provides a solid foundation for the use of \eqref{eq:realdecomp} in BD-RIS design and other problems involving unitary and symmetric matrices.}. Another way to prove it is as follows: given $\U \in \CU_s$, let us write its real and imaginary parts, $\U = \U_r + j \U_i$. It is easy to check from the facts that $\U\U^H=\I_n$ and $\U=\U^T$  that $\U_r$ and $\U_i$ commute and hence are simultaneusly diagonalizable:
   \begin{eqnarray*}
       \U_r &=& \V \diag(\cos(\theta_1), \ldots, \cos(\theta_n)) \V^T, \\
       \U_i &=& \V \diag(\sin(\theta_1), \ldots, \sin(\theta_n)) \V^T,
   \end{eqnarray*} 
    therefore yielding \eqref{eq:realdecomp}. Conversely, given $\U \in \CU_s$ factorized as \eqref{eq:realdecomp}, we can easily construct a unitary matrix $\Q = \V {\bf \Lambda}^{1/2}$ and hence $\U = \Q \Q^T$.

    \end{remark}

\section{Manifold optimization algorithms}
\label{sec:MOgeneric}
\subsection{LS and PO algorithms}
Suppose we want to maximize a continuous and differentiable function, $f(\U)$, with $\U \in \CU_s$, using the classical gradient ascent algorithm. It remains to be discussed how to select the step size, $\mu$, along the geodesic in \eqref{eq:geodesic2}. The standard approach to find a suitable $\mu$ is to apply a line search procedure that ensures $f(\U_{k+1})>f(\U_k)$, as described in \cite[Ch. 4]{AbsilBook}. Nevertheless, the parameterization of the geodesic in \eqref{eq:geodesic2} allows us to conceive other methods that increase the search space by optimizing a vector $\mub = (\mu_1,\ldots,\mu_n)$ instead of optimizing a single adaptation step $\mu$ common to all phases $\theta_m$. More specifically, we consider two different MO algorithms in $\CU_s$ that differ in how the phases in the diagonal matrix defined in \eqref{eq:RISmatrix} are chosen. 
\begin{itemize}
    \item {\bf Line Search (LS)}: The step size $\mu$ is common to all phases. The standard approach to find a suitable $\mu$ is to apply a line search procedure that ensures $f(\U_{k+1})>f(\U_k)$. Typically, it is relatively easy to find the optimal value of $\mu$ by bisection that maximizes the univariate function
\begin{equation}
   \max_{\mu \geq 0 } g(\mu) =  \max_{\mu \geq 0 } f(\Q_R  \Lambdab_R^\mu \Q_R^T),
   \label{eq:optmu}
\end{equation}
where $\Q_R$ and $\Lambdab_R$ are the values at the current iteration.
 \item {\bf Phase Optimization (PO)}: Each phase of the diagonal matrix in \eqref{eq:RISmatrix} can use a different $\mu_m$ parameter, thereby expanding the search space and potentially accelerating the algorithm. More specifically, the new phases, denoted as $\theta_m' = \theta_m \mu_m$, $m=1,\ldots, n$, can be optimized independently. Depending on the cost function, this optimization can be performed sequentially or, in some cases, jointly with closed-form solutions. The proposed algorithm\footnote{Matlab code can be downloaded from \url{https://github.com/IgnacioSantamaria/Code-MaxCap-BD-RIS-MIMO}.} with this phase optimization procedure for the selection of the new phases is summarized in Algorithm~\ref{alg:MOopt}.
\end{itemize}

\subsection{Convergence analysis}
In this subsection, we analyze the convergence of the two Riemannian optimization algorithms on $\CU_s$. We assume that the function to be maximized, $f(\U)$, is bounded from above, continuous, and differentiable. Regarding the LS-based MO algorithm, if the step size $\mu$ in the LS procedure is selected to satisfy $f(\U_{k}) > f(\U_{k-1})$  (chosen, for example, using Armijo’s rule) and under mild conditions (see \cite[Ch. 4]{AbsilBook}), the sequence of iterates $\{\U_k\}$, $k=0,1,\ldots,$ provided by Algorithm \ref{alg:MOopt} converges to a stationary point of $f(\U)$. A detailed proof of the global convergence of monotone line search Riemannian algorithms on manifolds can be found in \cite[Ch. 4]{AbsilBook}.

\begin{algorithm}[!t]
\small
\DontPrintSemicolon
\SetAlgoVlined
\KwIn{Initial $\U_0 = \Q_0\Q_0^T \in \CU_s$, Cost function: $F_0 = f(\U_0)$, convergence threshold: $\epsilon$}
\KwOut{Final $\U = \Q\Q^T \in \CU_s$}
\For{$k = 1,\ldots$}{
Calculate gradient $\J =\nabla f(\U)$ at $\U_{k-1} = \Q_{k-1}\Q_{k-1}^T$ \;
Find $\R$ as in \eqref{eq:proj}\;
Perform eigendecomposition $j\R = \V_R \Sigmab_R \V_R^T$\;
$\Lambdab_R = \exp(\Sigmab_R) = \diag \left(e^{j\theta_1},\ldots, e^{j\theta_n} \right)$\;
$\Q_{R} = \Q_{k-1} \V_R$ \;
\For{$m=1,\ldots,n$}{
$\theta_m' = {\rm arg}\max_{\theta_m} f(\Q_{R} \Lambdab_R(\theta_m) \Q_{R}^T)$}

Obtain the new Takagi factor as $\Q_{k} = \Q_{k-1} \V_R \Lambdab_R^{1/2}$\;
Obtain $\U_{k} = \Q_{k}\Q_{k}^T$ and $F_k = f(\U_{k})$\;
Check convergence: $|F_k-F_{k-1}| < \epsilon$ 
}
\caption{Manifold optimization algorithm in $\CU_s$ with Phase Optimization (for the LS version, step 7 is replaced by $\theta_m' = \mu \theta_m$ with $\mu$ selected by a line search procedure).}
\label{alg:MOopt}
\end{algorithm}

Regarding the PO-based MO algorithm, since the cost function cannot decrease with iterations, its monotonic convergence is guaranteed. Nevertheless, this is a trivial property that says little about the convergence of the sequence of iterates $\{\U_k\}$. To prove convergence to a stationary point, a more detailed analysis is required for each specific function to be optimized. The PO procedure in steps 7–8 of Algorithm \ref{alg:MOopt} belongs to the category of block coordinate ascent methods  \cite{Bertsekas_nonlinear}, where the objective function $f(\theta_1,\ldots,\theta_n)$ is cyclically maximized with respect to each phase $\theta_m \in [0, 2\pi]$. The following conditions are sufficient for first-order convergence (for a more detailed analysis, see \cite{Tseng01}, \cite{Luo13}):
\begin{itemize}
\item[C1)] {\bf Continuity}. $f(\theta_1,\ldots,\theta_n)$ is continuous on $[0,2\pi]^n$. Since this domain is compact, $f(\theta_1,\ldots,\theta_n)$ is bounded above and every sequence of iterates has at least one limit point.
\item[C2)] {\bf Differentiability}. $f(\theta_1,\ldots,\theta_n)$ is differentiable with respect to each $\theta_m$.
\item[C3)] {\bf Unique maximizer per phase}. For each $m$ and for all fixed values of the remaining phases, $f(\theta_1,\ldots, \theta_{m-1}, \theta_m, \theta_{m+1},\ldots,\theta_n)$ has a unique maximizer on $[0,2\pi]$. This is the critical condition which rules out the pathological cases discussed in \cite{Tseng01}.
\item[C4)] {\bf Cyclic update rule}. Every phase is updated at least once every $n$ steps.
\end{itemize}

The most restrictive and critical condition is (C3). In the application to BD-RIS design discussed in the following section, the three cost functions considered have a unique solution for each phase-optimization subproblem. Therefore, convergence of the PO algorithm to a stationary point is also guaranteed in those cases.

\section{Optimization problems in BD-RIS-assisted MIMO links}
\label{sec:applic}
\subsection{Introduction}
\label{sec:appintr}
The design of reconfigurable intelligent surfaces (RISs), and, in particular, of beyond-diagonal RIS (BD-RIS), is the application that has motivated the study of the manifold $\CU_s$. A BD-RIS is a two-dimensional structure composed of $M$ reconfigurable elements that allow the phase and amplitude of the incident electromagnetic wave to be manipulated. For example, the electromagnetic wave reflected by the surface can be directed toward an intended receiver to maximize the power of the received signal. The recent tutorial \cite{ClerkxBDRISTut26} provides a comprehensive analysis of BD-RISs, including their mathematical modeling, proposed architectures, the various metrics they can optimize, and their performance benefits. Specifically, a fully connected BD-RIS with $M$ elements is modeled by an $M \times M$ scattering matrix, $\Thetab$, all of whose entries can be optimized. Since it is modeled as a lossless device (i.e., all incident energy is reflected by the BD-RIS), $\Thetab$ is a unitary matrix. In addition, if the BD-RIS is implemented with reciprocal passive components such as capacitors or inductors, $\Thetab$ is a symmetric matrix. Under these assumptions, the scattering matrix $\Thetab$ of a fully connected BD-RIS belongs to $\CU_s$.

We consider the design of a fully connected BD-RIS with $M$ elements in a multiple-input multiple-output (MIMO) channel with $N_t$ transmit antennas and $N_r$ receive antennas. The equivalent MIMO channel is
\begin{equation}
\H_{eq} (\Thetab) = \H_d + \F \Thetab \G^H,
\label{eq:MIMOchanneleq}
\end{equation} 
where $\G \in \mathbb{C}^{N_t \times M}$ is the channel from the Tx to the BD-RIS, $\F \in \mathbb{C}^{N_r \times M}$ is the channel from the BD-RIS to the Rx, $\H_d \in \mathbb{C}^{N_r \times N_t}$ is the MIMO direct link, and $\Thetab \in \CU_s$ is the $M \times M $ BD-RIS matrix.

To optimize the usual metrics in communication problems, such as the sum channel gain (Frobenius norm of $\H_{eq} (\Thetab)$) or the achievable data rate, there are generally no closed-form solutions for the optimal $\Thetab \in \CU_s$. Existing solutions are either iterative, with varying degrees of computational complexity \cite{Zhou24, SantamariaSPAWC24}, or suboptimal, based on some relaxation of the BD-RIS constraints followed by a projection onto the set of unitary and symmetric matrices \cite{MaoCL2024}. Next, we first describe the BD-RIS-assisted MIMO simulation scenario and then apply the proposed MO framework in $\CU_s$ to find solutions to the maximization of the sum channel gain and the maximization of the achievable rate.

\subsection{Low-rank BD-RIS solution}
\subsubsection{Main results}
To increase channel shaping flexibility, it is common to use a BD-RIS with a large number of elements such that $M >> \max(N_t, N_r)$. Under this assumption, the channel through the BD-RIS is a product of a left fat matrix $\F$, the square BD-RIS matrix $\Thetab$, and a right tall matrix $\G^H$. This suggests that only the row space of $\F$ and the column space of $\G^H$ play a role in the optimization of $\Thetab$, and thus a low-rank $\Thetab$ might provide the optimal value of the cost function. This idea is formalized in the following result, which shows that for any cost function  the rank of the optimal solution is {\it at most} $r = N_t+N_r$. As we shall see, this result can be exploited to significantly reduce the computational cost of the corresponding MO algorithms.

\begin{theorem}
\label{th:theoremUs}
Let $\Thetab \in \CU_s$ be an arbitrary $M \times M$ full-rank unitary and symmetric BD-RIS matrix, and consider an optimization problem written as $\max_{\Thetab \in \CU_s} f(\F \Thetab \G^H)$, where $\F$ and $\G$ are defined as in \eqref{eq:MIMOchanneleq} and $M > N_r+N_t$. Then, the following symmetric BD-RIS matrix of rank at most $r = N_r+N_t$,
\begin{equation}
 \Thetab_{lr} = \U_Z \tilde{\Thetab} \U_Z^T,   
\end{equation}
where $\tilde{\Thetab} =\U_Z^H\Theta\U_Z^*$ and $\U_Z$ is a basis for the column space of the $M \times r$ matrix $\Z = \left [ \F^H | \G^T \right ]$ satisfies: i) $\F \Thetab\G^H = \F\Thetab_{lr}\G^H$ and hence $\H_{eq} (\Thetab) = \H_{eq} (\Thetab_{lr})$, and ii) $\Thetab_{lr} \Thetab_{lr}^H \preceq \I_M$. The structure of the low-rank matrix $\Thetab_{lr}$ is depicted in Fig. \ref{fig:StrThetaLRUs}, where we note that the inner block $\tilde{\Thetab}$ is an $r \times r$ symmetric matrix that satisfies $\tilde{\Thetab}\tilde{\Thetab}^H \preceq \I_{r}$. 
\end{theorem}

\begin{proof}
See Appendix B.
\end{proof}

\begin{figure}[t]
    \centering
    \begin{tikzpicture}[every node/.style={align=center}]
  \coordinate (top) at (0,0);
  \node (A) [draw,minimum height=3.1cm, minimum width=1.2cm, anchor=north west] at (top) {\(\U_Z\)};
  \node (B) [draw,minimum height=1.2cm, minimum width=1.2cm, anchor=north west] at ([xshift=1.2cm+0.1cm]top) {\(\tilde{\Thetab}\) };
  \node (C) [draw,minimum height=1.2cm, minimum width=3.1cm, anchor=north west] at ([xshift=1.2cm+0.1cm+1.3cm]top) {\(\U_Z^T\) };
  \node[anchor=center] at (-0.7,-1.6) {$\Thetab_{lr}=$};
\end{tikzpicture}
    \caption{Structure of the low-rank {\it symmetric} matrix, $\Thetab_{lr}$, in Theorem \ref{th:theoremUs}. The inner block, $\tilde{\Thetab}$, is an $r \times r$ symmetric matrix with $r = N_t+N_r$.}
    \label{fig:StrThetaLRUs}
\end{figure}
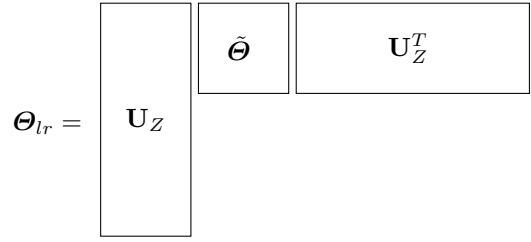

Suppose we have the following problem
\begin{equation}
     ({\cal{P}}_1): \quad \max_{\Thetab \in \CU_s(M)} f(\F \Thetab \G^H),
     \label{eq:P1}
\end{equation}
where $\F$, $\Thetab$, and $\G$ have the dimensions specified in \eqref{eq:MIMOchanneleq}. Compute a basis, $\U_Z$, for the column space of the $M \times r$ matrix $\Z = \left [ \F^H | \G^T \right ]$, and define the transformed channels $\tilde{\F} = \F \U_Z$ and $\tilde{\G} = \G \U_Z^*$. Notice that $\tilde{\F}$ is $N_r \times r$ and ${\tilde \G}$ is $N_t \times r$. Theorem \ref{th:theoremUs} shows that ${\cal{P}}_1$ is equivalent to
\begin{equation}
     ({\cal{P}}_2): \quad \max_{\tilde{\Thetab} \in {{\cal{B}}_s}(r)} f(\tilde{\F} \tilde{\Thetab} {\tilde \G}^H),
     \label{eq:P2}
\end{equation}
where ${{\cal{B}}_s}(r) = \{ \U \in \mathbb{C}^{r \times r}: \U = \U^T, \U\U^H \preceq \I_r \}$. For the problems considered in this paper (maximization of the Frobenius norm of the equivalent channel, maximization of the achievable rate, or minimization of the MSE), numerical checks show that the solution to ${\cal{P}}_2$ can be chosen on $\CU_s(r)$. This means that the inequality $\tilde{\Thetab} \tilde{\Thetab}^H \preceq \I_{r}$ is satisfied with equality $\tilde{\Thetab} \tilde{\Thetab}^H = \I_{r}$ at some maximizer (or minimizer for the MSE cost function). The following proposition formalizes this claim for the maximization of the Frobenius norm.

\begin{proposition}\label{prop:frobnorm} Suppose $M > r= N_t+N_r$ with $\tilde{\F}$ and $\tilde{\G}$ defined as before. Then, the following problem
\begin{equation}
({\cal{P}}_{2a}): \quad \max_{\tilde{\Thetab} \in {\cal{B}}_s(r)} \|\tilde{\F} \tilde{\Thetab} {\tilde \G}^H\|_F^2,
\label{eq:P2a}
\end{equation}  
is equivalent to
\begin{equation}
     ({\cal{P}}_{3a}): \quad \max_{\tilde{\Thetab} \in \CU_s(r)} \|\tilde{\F} \tilde{\Thetab} {\tilde \G}^H\|_F^2,
     \label{eq:P3a}
\end{equation}
where we have replaced the feasibility set ${\cal{B}}_s(r)$ by the manifold $\CU_s(r)$. In other words, the maximum of the Frobenius norm in \eqref{eq:P2a} is attained at a rank-$r$ unitary and symmetric matrix. 
\end{proposition}

\begin{proof}
See Appendix C.
\end{proof}     

\subsubsection{Implications of a low-rank solution}  
An $M \times M $ BD-RIS is implemented through an $M$-port network interconnected through a reconfigurable admittance matrix, $\Y$, which is the inverse of the impedance \cite{ClerckxTWC22a}, \cite{Nossek24}. To maximize the power reflected by the intelligent surface, the admittance matrix must be purely imaginary, i.e., $\Y = j\B$, where $\B$ is the so-called susceptance matrix \cite{ClerkxBDRISTut26},\cite{Nerini24}. The scattering matrix, $\Thetab$, and the susceptance matrix, $\B$, are related by the Cayley transform (cf. Sec.~\ref{sec:Cayley})
\begin{equation}
    \Thetab = (\I_M + jZ_0 \B)^{-1}(\I_M - jZ_0 \B),
    \label{eq:B2Thetainv}
\end{equation}
where $Z_0$ is the reference impedance, usually set to  $Z_0 =  50 \, \Omega$. The diagonal elements, $b_{ss}$, $s=1,\ldots, M$, of $\B$ are the inverses of the impedances connecting each reflecting element to ground, while $b_{st}$ ($s \neq t$) model the connection between elements $s$ and $t$. From \eqref{eq:B2Thetainv}, it follows that for any configuration of $\B$, $\Thetab$ is always a full-rank matrix. What, then, are the implications of finding a low-rank solution to our optimization problems?

One immediate advantage is a significant reduction in computational complexity, since the optimization problem can be solved in $\CU_s(r)$ rather than in $\CU_s(M)$. Notice that, from the $r \times r$ solution, $\tilde{\Thetab}$, we recover a low-rank $M \times M$ solution as $\Thetab_{lr} = \U_Z \tilde{\Thetab} \U_Z^T$, which can be factored as $\Thetab_{lr} = \Q\Q^T$, where $\Q$ is a semi-unitary $M \times r$ matrix. The low-rank solution can be completed to a full-rank $M \times M$ unitary and symmetric matrix in $\CU_s(M)$ as
\begin{equation}
    \Thetab_{fr} = \underbrace{\Q \Q^T}_{\Thetab_{lr}} + \Q_{\perp} \Q_{\perp}^T,
    \label{eq:fullrank}
\end{equation}
where $\Q_{\perp}$ is the orthogonal complement of $\Q$. Since $\F (\Q_{\perp} \Q_{\perp}^T) \G^H = {\bf 0}_{N_r \times N_t}$, $\Thetab_{fr}$ and $\Thetab_{lr}$ are equivalent solutions. In fact, we could rotate the subspace $\Q_{\perp}$ arbitrarily without altering the solution.

A second advantage of the low-rank solution is that, since the scattering matrix needs to be approximated only in an $r$-dimensional subspace, it is possible to find structured susceptance matrices, $\B$, that require fewer connections than the traditional fully-connected architecture \cite{ClerkxBDRISTut26,ClerckxTWC22b,ClerckxTWC22a}, in which all reflecting elements are interconnected with a total number of $M(M+1)/2$ reconfigurable circuits. In particular, the $q$-stem topology for $\B$, proposed in \cite{ICC_QStem} and analyzed in detail in \cite{Zhou2025Arxiv,WuTIT25,Nerini24}, is capable of perfectly approximating $\Thetab_{lr}$ and, in this way, achieving the performance of fully connected networks with a significantly lower number of impedances. The susceptance matrix of the $q$-stem topology is structured as
\begin{equation}
        \B = \begin{bmatrix} \B_{11} & \B_{12}^T \\ \B_{12}  & \Gammab \end{bmatrix},
        \label{eq:qstem}
\end{equation}
where $\B_{11} = \B_{11}^T$ is a $q \times q$  symmetric and real matrix, $\B_{12}$ is a $(M-q) \times q$ block and $\Gammab = \diag(b_{qq+1},\ldots b_{M})$ is a $(M-q) \times (M-q)$ diagonal matrix. Therefore, the lower-right off-diagonals of $\B$ in the $q$-stem topology are zero, thus reducing the total number of reconfigurable elements to $\nu = q(q+1)/2 + (M-q)(q+1)$. In \cite[Corollary 2]{WuTIT25} and \cite[Lemma3]{MaxDet}, it is shown that a rank-$r$ scattering matrix can be implemented by a $q$-stem topology with $q = r-1$ stems.

\subsection{Sum channel gain maximization}
\label{sec:FrobNormopt}
The sum channel gain maximization problem amounts to maximizing the Frobenius norm of the equivalent channel
\begin{equation}
\label{eq:MaxMIMOFrobNOrm}
\,\max_{\Thetab \in \CU_s(M)}\,\, \| \H_d + \F \Thetab \G^H \|_F^2.
\end{equation}
Disregarding scaling constants, the unconstrained gradient of \eqref{eq:MaxMIMOFrobNOrm} is
\begin{equation}
    \J = \F^H \left( \H_d + \F \Thetab \G^H \right) \G. 
    \label{eq:unconstrainedC1}
\end{equation} 
The projection of $\J$ onto the tangent plane gives us the symmetric and real matrix $\R$ in \eqref{eq:proj}. From the eigendecomposition of $\R$, the new BD-RIS can be parameterized as $\Thetab = \Q_R \diag \left(e^{j\theta_1},\ldots, e^{j\theta_M} \right) \Q_R^T$. For the LS procedure, we find the optimal step size by bisection. For the PO procedure, we substitute the BD-RIS parametrization $\Thetab = \Q_R \diag \left(e^{j\theta_1},\ldots, e^{j\theta_M} \right) \Q_R^T$ into the cost function \eqref{eq:MaxMIMOFrobNOrm} and apply an alternating optimization method to optimize each of the phases $\theta_m$ independently. More precisely, the equivalent channel can be written as a function solely of $\theta_m$ as
\begin{equation}
\H_{eq}(\theta_m) = \underbrace{\H_d + \F \Q_R \Lambdab_R(e^{j\theta_m} = 0)\Q_R^T \G^H}_{{\bf S}} + e^{j\theta_m} \f_m \g_m^H,  
\label{eq:Heqphase}
\end{equation}
where $ \Lambdab_R(e^{j\theta_m} = 0)$ denotes the diagonal matrix $\Lambdab_R = \diag \left(e^{j\theta_1},\ldots, e^{j\theta_M} \right)$ with the $m$-th entry set to zero, $\f_m$ is the $m$-th column of $ \F \Q_R$ and $\g_m^H$ is the $m$-th row of $\Q_R^T \G^H$. Assuming the other phases are fixed, the optimal $\theta_m$ that maximizes the Frobenius norm of $\H_{eq}(\theta_m)$ is
\begin{equation}
    \theta_m' = - \angle \g_m^H {\bf S}^H \f_m,
    \label{eq:phaseopt}
\end{equation}
where ${\bf S}$ is defined in \eqref{eq:Heqphase}. In an inner iteration of the MO algorithm with PO, we iteratively optimize the phases from $m = 1$ to $m = M$, as in \eqref{eq:phaseopt}.

To compare the performance of the different MO algorithms, we consider an $N_t \times N_r$ MIMO system with the Tx located at coordinates (0,0,1.5) [m] and the Rx at coordinates (50,0,1.5) [m]. A fully-connected BD-RIS with $M$ elements is located close to the Rx at coordinates (50,3,3) [m] to assist the Tx-Rx communication. The channels through the BD-RIS have a dominant line-of-sight path and are therefore modeled as Rician with factor $K=3$ and path loss exponent $\alpha = 2$. The direct channel is modeled as a Rayleigh channel with path loss exponent that we take as $\alpha = 3.75$. The bandwidth is $20$ MHz and the system operates at 2.4 GHz. The path loss is $PL = PL_0 -\alpha 10 \log_{10} d$ where  $PL_{0} = 28$ dB is the path loss at a reference distance of $d_0 = 1$ meter and $\alpha$ is the path loss exponent. The power spectral density for the additive noise is $\sigma_n^2 =-174 + 10\log_{10}B$ dBm, and the transmit power is $P= 100 $ mW. The results shown are the average of 200 independent simulations (channel realizations), keeping the Tx, Rx, and BD-RIS positions fixed. For the proposed MO algorithms (both for the LS and PO versions), we used a convergence threshold of $\epsilon = 10^{-3}$.

In this section, we compare the performance of the following algorithms:
\begin{itemize}
     \item {\bf Full-rank MO algorithm with LS procedure}: The proposed MO algorithm operating over $\CU_s(M)$ (full-rank version) using a LS procedure. 
     \item {\bf Full-rank MO algorithm with PO procedure}: The proposed MO algorithm operating over $\CU_s(M)$ (full-rank version) using an alternating PO procedure. 
     \item {\bf Low-rank MO algorithm with LS procedure}: The proposed MO algorithm operating over $\CU_s(r)$ (low-rank version) with $r=N_t+ N_r$ using a LS procedure. 
     \item {\bf Full-rank MO algorithm with PO procedure}: The proposed MO algorithm operating over  $\CU_s(r)$ (low-rank version) with $r=N_t+ N_r$ using an alternating PO procedure. 
\end{itemize}

Fig. \ref{fig:ConvergenceComp} shows the convergence curves of the four different versions of the MO algorithm that solves \eqref{eq:MaxMIMOFrobNOrm}. All versions converge very quickly in the first few iterations, and then the final convergence to the stationary point is slower. The PO versions are slightly faster than the LS versions. A more accurate comparison of the different MO algorithms requires an analysis of their computational complexity. We provide a detailed analysis of the computational cost of the full-rank and low-rank PO versions. The LS versions replace the phase optimization by a bisection procedure. To simplify the notation, we consider an $N \times N$ MIMO system. Therefore, for the low-rank version $r = 2N$.

\begin{itemize}
    \item {\bf Full-rank MO algorithm with PO procedure}. The computation of the gradient in \eqref{eq:unconstrainedC1}  has a complexity $\mathcal{O}(MN^2 + M^2N)$. The projection onto the tangent plane involves the multiplication of $M \times M$ matrices and is hence $\mathcal{O}(M^3)$, and the eigendecomposition of $\R$ is $\mathcal{O}(M^3)$. The cost to compute each phase as in \eqref{eq:phaseopt} is $\mathcal{O}(N^2)$ and for each round we need to update $M$ phases. Once we have updated all phases, the update of the Takagi factor to $\Q_k = \Q_{k-1} \V_R \Lambdab_R^{1/2}$ is $\mathcal{O}(M^3)$. Finally, the computational cost per iteration grows as $\mathcal{O}(M N^2+ M^2 N + M^3)$, which is dominated by the last cubic term.

     \item {\bf Low-rank MO algorithm with PO procedure}. Since we are now optimizing a $2N \times 2N$ matrix $\tilde{\Theta_b}$ (remember that $r = 2N$), the cost per iteration of the MO algorithm is $\mathcal{O}(N^3)$, which is typically orders of magnitude smaller than the cost of the full-rank version since $M \gg N$. In addition to the per-iteration complexity, we must add the cost of obtaining a basis for $\Z = \left [ \F^H | \G^T \right ]$, which is $\mathcal{O}(M N^2)$, and the cost of obtaining the channels $\tilde{\F} = \F \U_Z$ and $\tilde{\G} = \G \U_Z^*$, which is $\mathcal{O}(M N^2)$.
\end{itemize}

To analyze the cost of the MO algorithm with LS, we consider the full-rank case. Instead of the cost of updating the phases, a conventional backtracking/Armijo line search along the geodesic evaluates the derivative of the cost function with respect to $\mu$ at 3–10 trial values of $\mu$ (halving or doubling $\mu$ each trial). Each trial is $\mathcal{O}(M^3)$  
and, consequently, the computational cost of LS is higher that the PO versions.

Fig. \ref{fig:RunTime}, which depicts the runtime of each algorithm as a function of $M$, confirms these results. The LS versions are computationally more expensive than the corresponding PO versions. And, logically, the cost of the low-rank versions is much lower than that of the full rank versions. Further, the runtime of the low rank versions is practically independent of $M$. Since all MO algorithms reach a stationary point with almost identical value of the cost function, it is advisable to use the low-rank PO version. 

\begin{figure}[htb]
    \centering
\includegraphics[width=.47\textwidth]{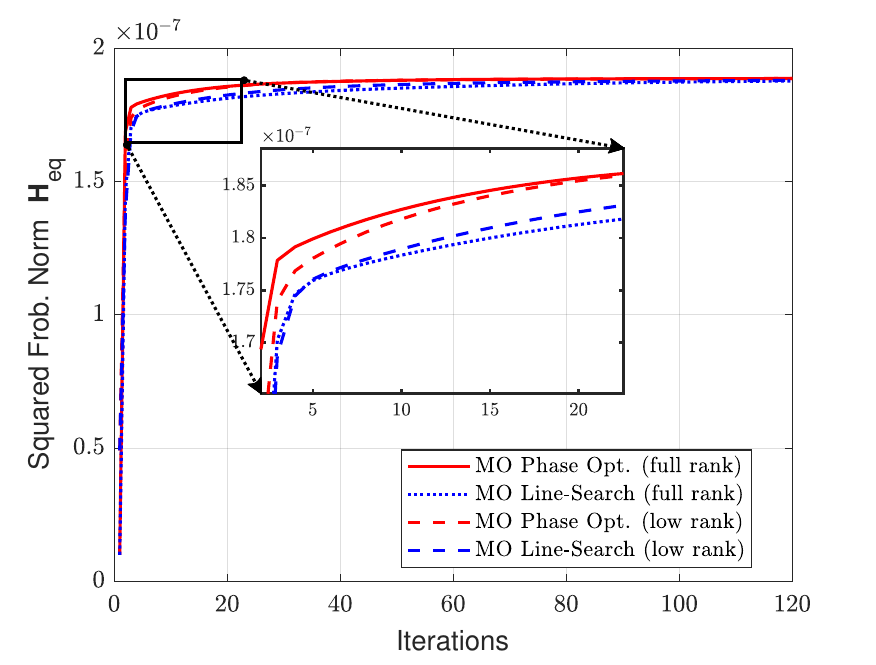}
     \caption{Convergence curves for the MO algorithms with line search (LS) and phase optimization (PO) for full-rank and low-rank BD-RIS.}
	\label{fig:ConvergenceComp}
\end{figure}

\begin{figure}[htb]
    \centering
\includegraphics[width=.47\textwidth]{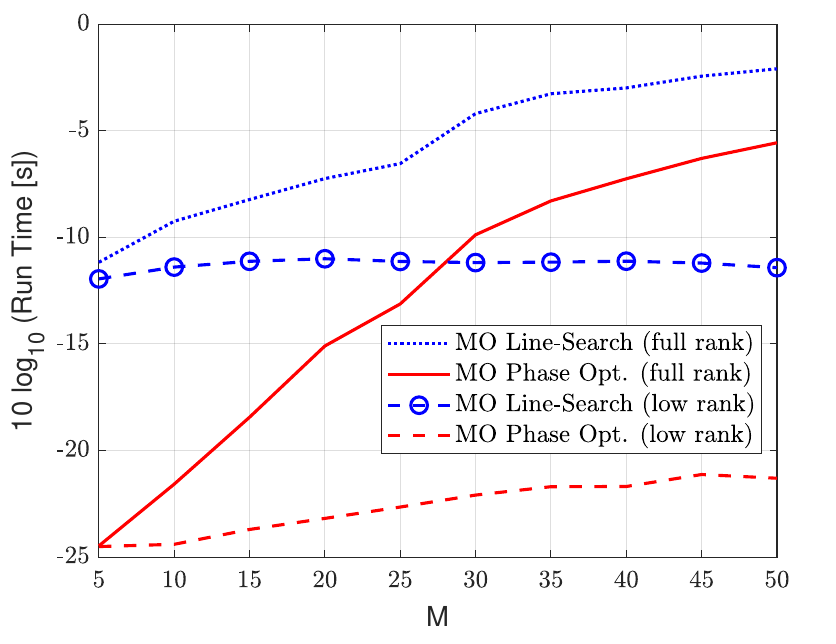}
     \caption{Run time in secs of the MO algorithms with line search and phase optimization.}
	\label{fig:RunTime}
\end{figure}

\subsection{Achievable rate maximization }
\label{sec:rateopt}
As a second application, we consider the maximization of the achievable rate in a BD-RIS-assisted MIMO system. We assume that the transmitter sends proper Gaussian signals with a fixed isotropic covariance matrix, $\x \sim \mathcal{CN}({\bf 0}, P \I)$, and the received signal is contaminated by additive white Gaussian noise, ${\bf n} \sim \mathcal{CN}({\bf 0}, \sigma^2 \I)$. Then, the BD-RIS that maximizes the achievable rate solves the problem
\begin{equation}
\label{eq:MaxCapMIMO}
\,\max_{\Thetab \in \CU_s(M)}\,\, \log \det \left( \I_{N_R} + \frac{P}{\sigma^2}\H_{eq}(\Thetab)\H_{eq}(\Thetab)^H \right), 
\end{equation}
where $\H_{eq}(\Thetab)$ is given by \eqref{eq:MIMOchanneleq}. Assuming a scenario in which $M > r= N_t+N_r$, we will consider only the low-rank PO version due to the advantages already mentioned in the previous application. First, we generate the new channels $\tilde{\F} = \F \U_Z$ and $\tilde{\G} = \G \U_Z^*$ of dimensions $N_r\times r$ and $N_t \times r$  where $\U_Z$ is a  basis for the column space of $\Z = \left [ \F^H | \G^T \right ]$. With this transformation, the new equivalent channel is $\H_{eq}(\tilde{\Thetab}) = \H_d + \tilde{\F} \tilde{\Thetab} \tilde{\G}^H$, but the problem is still posed as \eqref{eq:MaxCapMIMO} with $\tilde{\Thetab} \in \CU_s(r)$. Let us introduce the following matrix
\begin{equation}
    \E(\tilde{\Thetab}) = \I_{N_R} + \frac{P}{\sigma^2}\H_{eq}(\tilde{\Thetab})\H_{eq}(\tilde{\Thetab})^H. 
    \label{eq:MSEmatrix}
\end{equation}
With this definition and disregarding scaling constants, the unconstrained gradient of the achievable rate is \cite{PalomarTIT06}
\begin{equation}
    \J =  \tilde{\F}^H \E(\tilde{\Thetab})^{-1} \H_{eq}(\tilde{\Thetab}) \tilde{\G}, 
    \label{eq:unconstrained2}
\end{equation}
which is an $r \times r$ matrix. Projecting $\J$ onto the tangent plane we get the $r \times r$ matrix $\R$, with eigendecomposition $j\R = \V_R \diag(j\theta_i, \ldots, j\theta_r) \V_R^T$. Notice that in the PO procedure, only $r$ phases need to be optimized. Assuming the other phases are fixed, the optimal $\theta_m$ that maximizes the achievable rate is \cite[Sec. III.B]{ZhangCapacityJSAC2020}
\begin{equation}
    \theta_m' = \angle \tilde{\f}_m^H {\bf A}^{-1}{\bf S}^H \tilde{\g}_m,
    \label{eq:phaseopt2}
\end{equation}
where ${\bf S}$ is defined in \eqref{eq:Heqphase} with $\tilde{\F}$ and $\tilde{\G}$ substituting $\F$ and $\G$, respectively; and $\A = \I_{N_r} + {\bf S}{\bf S}^H + \tilde{\f}_m \tilde{\f}_m^H \|\tilde{\g}_m \|_2^2 $.

We consider the same scenario as in the previous application, including the possibility that the direct channel is blocked. To do this, the path loss exponent of the direct channel is set to $\alpha = 8$. We compare the performance of the proposed low-rank MO algorithm with PO (labeled as {\bf MO+PO in $\CU_s(r)$}) with the following methods:
\begin{itemize}
     \item {\bf MM+MO in $\CU(M)$}: This method \cite{SantamariaSPAWC24} applies a mino\-rize-maximi\-zation (MM) algorithm and exploits the Takagi factorization of the BD-RIS. 
    \item {\bf MO in $\CU(M)$ + Proj}: it applies an MO algorithm on $\CU(M)$ to find a unitary but non-symmetric BD-RIS, $\Thetab_u$, that maximizes the rate. The final BD-RIS is obtained as the retraction (see Proposition \ref{prop:retraction}) $\Pi(\Thetab_u + \Thetab_u^T)$.
    \item {\bf Low Cost}: The BD-RIS is obtained as $\Pi(\F^H\H_d\G + (\F^H\H_d\G)^T)$, as proposed in \cite{MaoCL2024}. It assumes that the direct link is not blocked.
\end{itemize}

\begin{figure}[htb]
    \centering
\includegraphics[width=.47\textwidth]{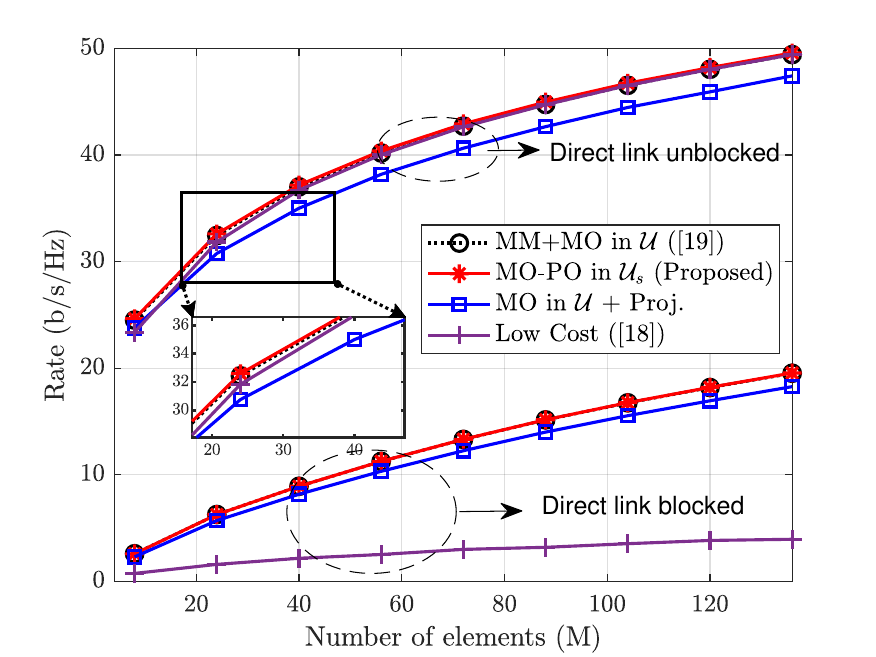}
     \caption{Rate vs. number of BD-RIS elements achieved with different algorithms.}
	\label{fig:RatevsM}
\end{figure}

Fig. \ref{fig:RatevsM} shows the rate versus the number of BD-RIS elements, $M$, when the direct channel is blocked (lower curves) and when there is a direct channel (upper curves). The proposed MO algorithm provides practically identical results to \cite{SantamariaSPAWC24} in both cases, although with much less computational effort. The results are also very similar to those provided by \cite{MaoCL2024} when the direct channel is unblocked, but when the direct channel is blocked, this method is not useful. Finally, the final retraction step after optimizing on $\CU(M)$ causes a significant rate penalty.

\subsection{MSE minimization}
As a final application demonstrating the generality of the proposed MO framework, we consider minimizing the mean squared error (MSE) between the transmitted and estimated signals at the Rx. In the RIS design literature, MSE minimization has been considered less frequently than sum-channel-gain or rate-maximization problems. RIS minimization with diagonal RIS has been considered in \cite{LongTWC24}, but to the best of our knowledge, it has not been used to optimize BD-RIS. As shown in \cite{PalomarFnT}, when the linear minimum MSE filter is applied at the Rx, the MSE matrix is precisely \eqref{eq:MSEmatrix}. Therefore, the optimization problem is
\begin{equation}
\label{eq:MinMSEMIMO}
\,\min_{\tilde{\Thetab} \in \CU_s(r)}\,\, \tr \left( \E(\tilde{\Thetab})^{-1} \right), 
\end{equation}
where $\E(\tilde{\Thetab})$ is the MSE matrix in \eqref{eq:MSEmatrix}. Assuming $M>r = N_t+N_r$, \eqref{eq:MinMSEMIMO} poses the low-rank optimization problem with $\H_{eq}(\tilde{\Thetab}) = \H_d + \tilde{\F} \tilde{\Thetab} \tilde{\G}^H$. To specialize our MO algorithm to MSE minimization, we only need the unconstrained gradient of \eqref{eq:MinMSEMIMO}, which, disregarding scaling constants, is
\begin{equation}
    \J =  \tilde{\F}^H \E(\tilde{\Thetab})^{-2} \H_{eq}(\tilde{\Thetab}) \tilde{\G}, 
    \label{eq:unconstrained3}
\end{equation}
where we see that the only difference with respect to the rate gradient in \eqref{eq:unconstrained2} is that the inverse of $\E(\tilde{\Thetab})$ is now squared. To avoid unnecessary repetition, we do not detail the calculation of the Riemannian gradient, the $\R$ matrix, and its eigendecomposition. Unfortunately, for the MSE minimization, there seems to be no closed-form solution for optimizing the phases in the MO-PO version. Therefore, for this optimization step, we use a standard golden-section search method.

Fig. \ref{fig:MSEvsM} compares the MSE (in log scale) obtained in a $2 \times 2$ MIMO system assisted by BD-RIS with an increasing number of elements $M$. We consider only the low-rank version with phase optimization. For comparison, we include in Fig. \ref{fig:MSEvsM} the MSE of the max-rate solution, as described in subsection \ref{sec:rateopt}. The min-MSE BD-RIS reduces the MSE by slightly more than 3 dB compared to the max-rate BD-RIS.  

\begin{figure}[htb]
    \centering
\includegraphics[width=.47\textwidth]{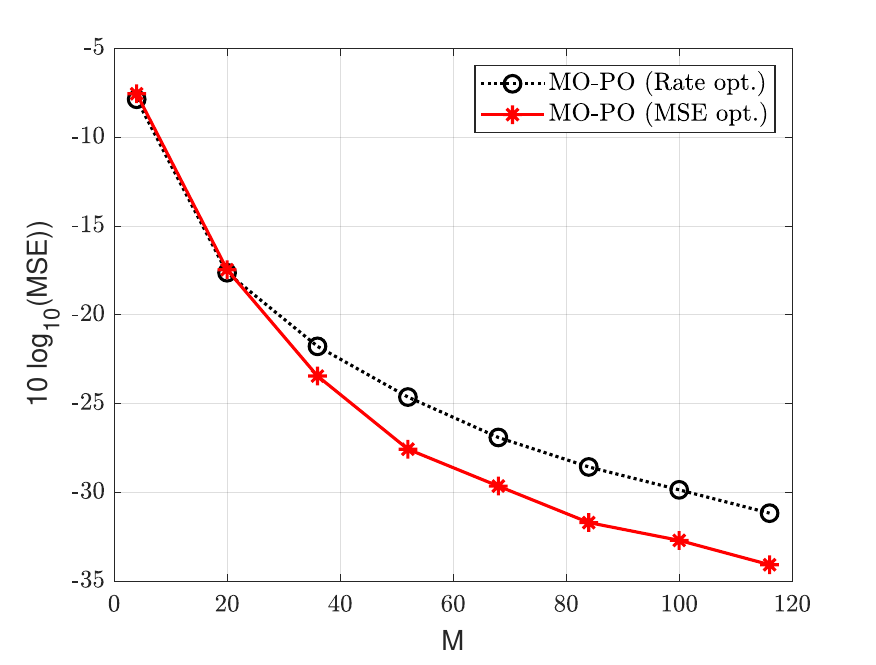}
     \caption{MSE vs. number of BD-RIS elements for MO algorithms optimizing rate or MSE.}
	\label{fig:MSEvsM}
\end{figure}

\label{sec:MSEopt}
\section{Conclusion}
\label{sec:conclusion}
In conclusion, the manifold of unitary and symmetric matrices has been characterized through the parameterization of its tangent space and geodesics, which has enabled the development of new Riemannian optimization algorithms for this manifold. Two variants stand out: one based on optimizing the adaptation step size via line search (LS), and the other, named phase optimization (PO) and based on the sequential optimization of the phases arising from the parameterization of the geodesic. Of particular interest is the PO version, which, although it does not have the convergence guarantees of the more classical LS version, in practice it seems to be computationally more efficient, especially in problems where closed-form formulas exist for phase optimization. As an application that motivates our study, we have considered the design of a BD-RIS in a MIMO channel to: i) maximize the Frobenius norm of the equivalent channel, ii) maximize the achievable data rate, or iii) minimize the MSE. For all cost functions, the MO-PO algorithm, particularly its low-rank version, stands out due to its fast convergence and optimal performance. In future work, we will consider the extension of the proposed framework to multiuser MIMO scenarios assisted by BD-RIS, as well as the direct optimization of the susceptance matrix (the imaginary part of the admittance matrix), which is related to the scattering BD-RIS matrix via the Cayley transform.


\section*{Appendix A: Proof of Proposition 1}
\label{AppendixA}
\noindent{\bf Two vector spaces:}
 Let  $\U\in\CU_s$ have Takagi decomposition $\U=\Q\Q^T$. Consider the following vector spaces:
\begin{align*}
	\mathcal S_1(\U) =& \{\B\in  \mathbb{C}^{n\times n}: \U^H\B+\B^H\U={\bf 0},\;\B=\B^T\},\\
    \mathcal S_2(\U)=& \{j\Q\R\Q^T:\R\in \mathbb{R}^{n\times n}, \R=\R^T \}.	
\end{align*}
We prove that they are indeed equal:
\begin{itemize}
    \item $\mathcal S_1(\U)\subseteq \mathcal S_2(\U)$: let $\B\in\mathcal S_1(\U)$ and consider $\R=-j\Q^H\B\Q^*$. Then, $\R^T=(-j\Q^H\B\Q^*)^T=(-j\Q^H\B^T\Q^*)=(-j\Q^H\B\Q^*)=\R$, hence $\R=\R^T$ is a symmetric matrix. Moreover, $\U=\Q\Q^T$ implies $\Q^*=\U^H\Q, \Q^H\U=\Q^T$ and hence:
\begin{eqnarray*}
    {\bf 0}& = & \U^H(-j\Q\R\Q^T)+(-j\Q\R\Q^T)^H\U \\
    &=& -j\Q^*\R\Q^T+j\Q^*\R^*\Q^T,
\end{eqnarray*}

that is to say, $\Q^*\R\Q^T=\Q^*\R^*\Q^T$, that readily implies $\R=\R^*$, that is $\R$ is a real matrix and hence we conclude that $\B=j\Q\R\Q^T\in\mathcal S_2(\U)$.
\item $\mathcal S_2(\U)\subseteq \mathcal S_1(\U)$: Let $\R$ be real and symmetric. Let $\B=j\Q\R\Q^T$. Then, $\B=\B^T$ and
\begin{eqnarray*}
\U^H\B+\B^H\U &=&j\U^H\Q\R\Q^T-j\Q^*\R\Q^H\U \\
&=&j\Q^*\R\Q^T-j\Q^*\R\Q^T=0,
\end{eqnarray*}

that is $\B\in\mathcal S_1(\U)$, and we are done.
\end{itemize}
The real dimension of both spaces is then equal, and easy to compute for the real dimension of $\mathcal S_2$ equals that of the set of real $n\times n$ symmetric matrices, that is, $n(n+1)/2$.

\mbox{ }\\

\noindent{\bf Restriction of the exponential map in $\CU$}: The exponential map $\exp_\U:\{\B\in\mathbb{C}^{n\times n}:\U^H\B+\B^H\U=0,\|\B\|_F<1/2\}\to\CU$ defines a diffeomorphism onto an open neighborhood of $\U$ in $\CU$ (see for example \cite{HallLieGroups}). We claim that, if additionally $\U\in\CU_s$, the image of the restriction
\[
\varphi_\U=\exp_{\U}\mid_{\{\B\in\mathcal S_1(\U),\|\B\|_F<1/2\}}
\]
is contained in $\CU_s$ and contains an open neighborhood of $\U$ in $\CU_s$:
\begin{itemize}
    \item First, we prove that the image of the restriction is contained in $\CU_s$. Let $\B\in\mathcal{S}_1=\mathcal{S}_2$, hence of the form $\B=j\Q\R\Q^T$ with $\R$ real and symmetric. Then,
\begin{eqnarray*}
   \exp_\U(\B)&=&\U e^{\U^H\B}=\U e^{\U^Hj\Q\R\Q^T} \\
   &=&\U e^{j\Q^*\R\Q^T}=\U\Q^* e^{j\R}\Q^T \\
   &=&\Q e^{j\R}\Q^T 
\end{eqnarray*}
is symmetric, since $\R$ is symmetric. Since it is also unitary, we conclude that
$\exp_\U(\B)\in\CU_s$.
\item The image of $\varphi_\U$ contains an open subset of $\CU_s$ around $\U$: let $\hat \U\in \mathcal U_s$ be sufficiently close to $\hat \U$. Since $\exp$ is a diffeomorphism to the space of unitary matrices, we have $\hat \U=\U\exp(\U^H\B)$ for some $\B$, $\U^H\B+\B^H\U={\bf 0}$, $\|\B\|_F<1/2$. We have:
\begin{eqnarray*}
e^{\B\U^H}&=&\U e^{\U^H\B}\U^H=\hat\U\U^H \\
&=&\hat \U^T\U^H=e^{(\U^H\B)^T}\U\U^H=e^{-(\B^H\U)^T} \\
&=&e^{-\U\B^*}.    
\end{eqnarray*}
Since $\|\B\|_F<1/2$ and the exponential is a diffeomorphism in this range, we conclude that $\B\U^*+\U\B^*=\B\U^H+\U\B^*=0$, which transposing implies
\[
\U^H\B^T=-\B^H\U^T=-\B^H\U=\U^H\B,
\]
which implies that $\B$ is symmetric. We have proved that the image of $\varphi_\U$ covers an open subset of $\mathcal U_s$.
\end{itemize}

\mbox{ }\\

\noindent{\bf $\mathcal U_s$ is a manifold}: The previous point shows that $\mathcal U_s$ satisfies the submanifold property: at every point $\U\in\CU_s$, there exists a coordinate system (given by the matrix exponential) of $\CU$ such that the image intersected with $\CU_s$ equals the image of a $n(n+1)/2$ real dimensional subspace $\mathcal S_1(\U)=\mathcal S_2(\U)$. The tangent space to $\CU_s$ at $\U$ is then precisely $\mathcal S_1(\U)$ and the exponential map in $\CU_s$ is just the restriction of the exponential map in $\CU$ to $\T_\U\CU_s$.

\section*{Appendix B: Proof of Theorem \ref{th:theoremUs}}
\label{AppendixB}
Since $\Thetab$ is symmetric and unitary, it admits a Takagi decomposition of the form $\Thetab = \Q\Q^T$, with $\Q$ unitary. Therefore, the left and right subspaces of the low-rank matrix must be matched to both $\F^H$ and $\G^T$. More precisely, let us define the $M \times r$ matrix $\Z = \left [ \F^H | \G^T \right ]$ and perform its SVD \cite[App. C]{Coherence} as
\[
\Z = \begin{bmatrix} \U_Z & \U_Z^{\perp} \end{bmatrix} \begin{bmatrix} \Sigmab_Z & {\bf 0} \\ 
 {\bf 0} & {\bf 0} \end{bmatrix} \begin{bmatrix} \V_Z^H \\ (\V_Z^{\perp})^H \end{bmatrix}. 
\]
Then, ${\bf P}_Z = \U_Z \U_Z^H$ is the orthogonal projection matrix onto the column space of $\Z = \left [ \F^H | \G^T \right ]$. It is a rank at most $r$, with $r = N_t+N_r$, idempotent matrix that satisfies $\F{\bf P}_Z = \F$,  ${\bf P}_Z^*= {\bf P}_Z^T$, and ${\bf P}_Z^T\G^H = \G^H$. Then, the matrix $\Thetab_{lr} = {\bf P}_Z \Thetab {\bf P}_Z^T$ is a {\it symmetric} rank at most $r$ matrix that satisfies 
\[
\F \Thetab_{lr} \G^H = \F {\bf P}_Z \Thetab {\bf P}_Z^T \G^H = \F \Thetab \G^H.
\]
To show that $\I_M - \Thetab_{lr} \Thetab_{lr}^H$ is positive semidefinite, it suffices to prove that the operator norm of $\Thetab_{lr}$ (with respect to the Euclidean norm) is at most 1, i.e., $\|\Thetab_{lr} \x\| \leq \|\x\|$ for all $ \x \in \mathbb{C}^M$. We have
\begin{equation*}
     \|\Thetab_{lr} \x\| = \|{\bf P}_Z \Thetab ({\bf P}_Z^T \x)\| \stackrel{(a)}{\leq} \|\Thetab ({\bf P}_Z^T \x)\|  \stackrel{(b)}{=} \|({\bf P}_Z^T \x)\|  \stackrel{(c)}{\leq} \| \x\|, \label{eq:proofnorm}
\end{equation*}
where $(a)$ follows since $\|{\bf P}_Z \| \leq 1$,  $(b)$ since $\Thetab\in\CU$, and $(c)$ since $\|{\bf P}_Z^T \| \leq 1$. Thus, the operator norm of $\Thetab_{lr}$ is at most 1 and $\I_M - \Thetab_{lr} \Thetab_{lr}^H$ is positive semidefinite.
\section*{Appendix C: Proof of Proposition \ref{prop:frobnorm}}
\label{AppendixC}
Let us consider the following problem
\begin{equation*}
({\cal{P}}_{2a}): \quad \max_{\tilde{\Thetab} \in {\cal{B}}_s(r)} \|\tilde{\F} \tilde{\Thetab} {\tilde \G}^H\|_F^2,
\end{equation*}  
and denote its optimal solution as $\tilde{\Thetab}_{opt}$. We want to prove that one can choose $\tilde{\Thetab}_{opt} \in \CU_s(r)$. The solution $\tilde{\Thetab}_{opt}$ is a symmetric matrix and hence has a Takagi factorization $\tilde{\Thetab}_{opt}=\Q{\bf\Sigma}\Q^T$ with $\Q$ unitary and $\Sigmab= \diag(s_1,\ldots,s_r)$, $s_m\in[0,1]$ for all $m$. Let us assume that for some $m$ we have $s_m<1$ and, consequently, $\tilde{\Thetab}_{opt} \notin \CU_s(r)$. Then, consider the perturbation of $\tilde{\Thetab}_{opt}$ given by
     \begin{equation*}\label{eq:perturbation}
         \tilde\Thetab_{pert}=\tilde\Thetab_{opt}  + se^{i\theta}\Q{\bf e}_m{\bf e}_m^T\Q^T=\Q(\Sigmab+se^{i\theta}{\bf e}_m{\bf e}_m^T)\Q^T,
     \end{equation*}
     with ${\bf e}_m$ the $m$--th vector of the standard basis in $\mathbb{C}^r$, $s\in\mathbb R$ and $\theta\in[0,2\pi)$ an angle to be fixed later. Then, we have
    {\small{
     \begin{equation*}
    \left. \frac{d \|\tilde{\F} \tilde\Thetab_{pert} {\tilde \G}^H\|_F^2}{ds} \right |_{s=0} = 2\Re\left(e^{-i\theta}{\bf e}_m^T\Q^H\tilde{\F}^H\tilde{\F} \tilde\Thetab {\tilde \G}^H\tilde\G\Q^*{\bf e}_m\right).
     \end{equation*}}}
If ${\bf e}_m^T\Q^H\tilde{\F}^H\tilde{\F} \tilde\Thetab_{opt} {\tilde \G}^H\tilde\G\Q^*{\bf e}_m\neq0$, choosing $\theta$ such that the real part in the right--hand side is positive and letting $s>0$ be small enough, we see that $\tilde\Thetab_{pert}\in {\cal{B}}_s(r)$ satisfies  $\|\tilde{\F} \tilde{\Thetab}_{pert} {\tilde \G}^H\|_F^2>\|\tilde{\F} \tilde{\Thetab}_{opt} {\tilde \G}^H\|_F^2$. This contradicts that $\tilde\Thetab_{opt}$ is optimal, hence we conclude that ${\bf e}_m^T\Q^H\tilde{\F}^H\tilde{\F}  \tilde\Thetab_{opt} {\tilde \G}^H\tilde\G\Q^*{\bf e}_m=0$. Then, we have
         \[
          \|\tilde{\F} \tilde\Thetab_{pert} {\tilde \G}^H\|_F^2=\|\tilde{\F} \tilde\Thetab_{opt} {\tilde \G}^H\|_F^2+s^2\|\tilde{\F}\Q{\bf e}_m{\bf e}_m^T\Q^T\tilde{\G}^H\|_F^2.
         \]
         As before, if $\|\tilde{\F}\Q{\bf e}_m{\bf e}_m^T\Q^T\tilde{\G}^H\|_F^2\neq0$ we get to a contradiction, hence we conclude that $ \|\tilde{\F} \tilde\Thetab_{pert} {\tilde \G}^H\|_F^2=\|\tilde{\F} \tilde\Thetab_{opt} {\tilde \G}^H\|_F^2$ for all $s$ and in particular we can choose $s=1-s_m$. This shows that a maximizer in $\CU_s(r)$ exists. In conclusion, Problem ${\cal{P}}_{2a}$ is equivalent to
\begin{equation*}
     ({\cal{P}}_{3a}): \quad \max_{\tilde{\Thetab} \in \CU_s(r)} \|\tilde{\F} \tilde{\Thetab} {\tilde \G}^H\|_F^2,
\end{equation*}
as claimed.

\bibliographystyle{IEEEtran}
\bibliography{refs}

@book{Coherence,
  title={Coherence: In Signal Processing and Machine Learning},
  author={Ramírez, D. and Santamaría, I. and Scharf, L.},
  year={2023},
  publisher={Springer Nature}
}

@article{SantamariaOJVT25,
  title={Interference Minimization in Beyond-Diagonal {RIS}-Assisted {MIMO} Interference Channels},
  author={Santamaria, I. and Soleymani, M. and Jorswieck, E. and Gutierrez,  J.},
  journal={IEEE Open Journal of Vehicular Technology},
  year={2025},
  volume={6},
  pages={1005-1017},
  publisher={IEEE}
}

@inproceedings{SantamariaSPAWC24,
  title={{MIMO} Capacity Maximization with Beyond-Diagonal {RIS}},
  author={Santamaria, I. and Soleymani, M. and Jorswieck, E. and Guti\'errez, J.},
  booktitle={IEEE International Workshop on Signal Processing Advances in Wireless Communications  (SPAWC)},
address = {Lucca, Italy},
pages={936--940},
 doi       = {10.1109/SPAWC60668.2024.10694491},
  year={2024}
}

@article{Autonne,
  title={Sur les matrices hypohermitiennes et sur les matrices unitaires},
  author={L. Autonne},
  journal={Ann. Univ. Lyon},
  volume={38},
  pages={1--77},
  year={1915}
}

@article{Takagi,
  title={On an algebraic problem related to an analytic theorem of {C}arathéodory and {F}ejér and on an allied theorem of {L}andau},
  author={T. Takagi},
  journal={Jpn. J. Math. Trans. Abstr.},
  volume={1},
  pages={83-93},
  year={1924}
}

@Book{boumal2023intromanifolds,
  title     = {An Introduction to Optimization on Smooth Manifolds},
  author    = {Boumal, N.},
  publisher = {Cambridge University Press},
  year      = {2023},
  url       = {https://www.nicolasboumal.net/book},
  doi       = {10.1017/9781009166164}
}

@article{SantamariaSPLetters2023,
  title={{SNR} Maximization in Beyond Diagonal {RIS}-assisted Single and Multiple Antenna Links},
  author={Santamaria, I. and Soleymani,  M. and  Jorswieck, E. and Guti\'errez, J.},
  journal={IEEE Signal Proc. Letters},
  volume={30},
  pages={923--926},
  year={2023},
  publisher={IEEE}
}

@article{edelman1998geometry,
  title={The geometry of algorithms with orthogonality constraints},
  author={Edelman, A. and Arias, T. A. and Smith, S. T.},
  journal={SIAM Journal on Matrix Analysis and Applications},
  volume={20},
  number={2},
  pages={303--353},
  year={1998},
  publisher={SIAM}
}

@article{MaoCL2024,
  title={A low-complexity beamforming design for beyond-diagonal {RIS} aided multi-user networks},
  author={Fang, T. and Mao, Y.},
  journal={IEEE Comm. Letters},
  volume={28},
  number={1},
  pages={203--207},
  year={2024},
  publisher={IEEE}
}

@inproceedings{EmilEuCap2025,
  title={Capacity maximization for {MIMO} channels assisted by beyond-diagonal {RIS}},
  author={Bj{\"o}rnson, E. and Demir, T.},
  booktitle={19th European Conference on Antennas and Propagation},
address = {Stockholm, Sweden},
  year={2025}
}

@article{ZhangCapacityJSAC2020,
  title={Capacity characterization for intelligent reflecting
surface aided {MIMO} communication},
  author={Zhang, S. and Zhang, R. },
  journal={IEEE J. Sel. Areas Commun.},
  volume={38},
  number={8},
  pages={1823--1838},
  year={2020},
  publisher={IEEE}
}

@article{NeriniTWC2023,
author = {Nerini, M. and Shen, S. and Clerckx, B.},
journal = {IEEE Trans. Wireless Commun.},
title = {Closed-form global optimization of beyond diagonal reconfigurable intelligent surfaces},
volume={23},
  number={2},
  pages={1037--1051},
  year={2024}
}

@article{PalomarTIT06,
  title={Gradient of mutual information in linear vector {G}aussian channels},
  author={Palomar, D. P. and Verdú, S. },
  journal={IEEE Trans. Inf. Theory},
  volume={52},
  number={1},
  pages={141--154},
  year={2006},
  publisher={IEEE}
}

@article{ClerckxTWC22a,
  title={Modeling and architecture design of reconfigurable intelligent surfaces using scattering parameter network analysis},
  author={Shen, S. and Clerckx, B. and Murch, R.},
  journal={IEEE Trans. Wireless Commun.},
  volume={21},
  pages={1229--1243},
  year={2022},
  publisher={IEEE}
}

@article{ClerckxTWC22b,
  title={Beyond diagonal reconfigurable intelligent surfaces: From Transmitting and reflecting modes to single-, group-, and fully-connected architectures},
  author={Li, H. and Shen, S. and Clerckx, B.},
  journal={IEEE Trans. Wireless Commun.},
  volume={22},
  number={4},
  pages={2311--2324},
  year={2022},
  publisher={IEEE}
}

@Book{HallLieGroups,
  title     = {Lie Groups, Lie Algebras, and Representations: An Elementary Introduction},
  author    = {Hall, B. C.},
  publisher = {Springer},
  year      = {2015}
}

@article{Xia25,
  title={Statistical {CSI}-Enabled Optimization for Beyond Diagonal {SIM}},
  author={Xia, Q. and Zhang, J. and Xu, K. and Ma, S. and Jin, S. and Yuen, C.},
  journal={IEEE Wireless Comm. Letters},
  volume={14},
  number={9},
  pages={2972--2976},
 doi= {10.1109/LWC.2025.3584817},
  year={2025},
  publisher={IEEE}
}

@book{AbsilBook,
	author = {P. A. Absil and R. Mahony and R. Sepulchre},
	publisher = {Princeton University Press},
	title = {Optimization Algorithms on Matrix Manifolds},
	year = {2008}}

@article{LongTWC24,
  title={{MMSE} Design of {RIS}-Aided Communications With Spatially-Correlated Channels and Electromagnetic Interference},
  author={Long, W. X. and Moretti, M. and Abrardo, A. and Sanguinetti, L. and Chen, R.},
  journal={IEEE Trans. Wireless Commun.},
  volume={23},
  number={11},
  pages={16992--17006},
  year={2024},
  publisher={IEEE}
}

@article{SunTSP24,
  title={Precoder Design for Massive {MIMO} Downlink With Matrix Manifold Optimization},
  author={Sun, R. and Wang, C. and Lu, A. and Gao, X. and Xia, X.-G.},
  journal={IEEE Trans. Signal Process.},
  volume={72},
  pages={1065--1080},
  year={2024},
  publisher={IEEE}
}

@inproceedings{ProjUNN,
author = {Kiani, B. T. and Balestriero, R. and LeCun, Y. and Lloyd, S.},
title = {proj{UNN}: efficient method for training deep networks with unitary matrices},
year = {2022},
isbn = {9781713871088},
address = {Red Hook, NY, USA},
booktitle = {Proceedings of the 36th International Conference on Neural Information Processing Systems},
articleno = {1050},
numpages = {16},
location = {New Orleans, LA, USA},
series = {NIPS '22}
}

@article{AbrudanTSP08,
  title={Steepest Descent Algorithms for Optimization Under Unitary Matrix Constraint},
  author={Abrudan, T. E. and Eriksson, J. and Koivunen, V.},
  journal={IEEE Trans. Signal Process.},
  volume={56}, 
  number={3},
  pages={1134--1147},
  year={2008},
  publisher={IEEE}
}

@inproceedings{santamariaICASSP26,
  title={Riemannian optimization on the manifold of unitary and symmetric matrices with application to {BD-RIS}-assisted systems},
  author={Santamaria, I. and Soleymani, M. and Jorswieck, E. and Gutierrez, J. and Beltran, C.},
  booktitle={IEEE International Conference on Acoustics, Speech and Signal Processing (ICASSP)},
  year={2026},
  address = {Barcelona, Spain}
}

@article{Zhou2025Arxiv,
  title={Generalized beyond-diagonal {RIS} architectures: Theory and design via structured-oriented symmetric unitary projection},
  author={Zhou, X. and Fang, T. and Mao, Y. and Clerckx, B.},
  journal={arXiv: https://arxiv.org/pdf/2509.17804v2},
  year={2025}
}

@article{WuTIT25,
  title={Beyond-Diagonal {RIS} in Multiuser {MIMO}: Graph Theoretic Modeling and Optimal Architectures With Low Complexity},
  author={Wu, Z. and Clerckx, B.},
  journal={IEEE Trans. on Inf. Theory},
  volume={71}, 
  number={11},
  pages={8506-8523},
  year={2025},
  publisher={IEEE}
}

@inproceedings{ICC_QStem,
author = {Zhou, X. and Fang, T. and Mao, Y.},
title = {A Novel {Q}-Stem Connected Architecture for Beyond-Diagonal Reconfigurable Intelligent Surfaces},
year = {2025},
booktitle = {IEEE International Conference on Communications (ICC)},
pages={6880-6885},
location = {Montreal, QC, Canada}
}

@article{ClerkxBDRISTut26,
    author = {Li, H. and Nerini, M. and Shen, S. and Clerckx, B.},
    title = {A Tutorial on Beyond-Diagonal Reconfigurable Intelligent Surfaces: Modeling, Architectures, System Design and Optimization, and Applications},
    journal = {IEEE Communications Surveys \& Tutorials},
    volume = {28},
    pages = {4086-4126},
    year = {2026}
}

@article{Zhou24,
    author = {Zhou, Y. and Liu, Y. and Li, H. and Wu, Q. and Shen, S. and Clerckx, B.},
    title = {Optimizing Power Consumption, Energy Efficiency, and Sum-Rate Using Beyond Diagonal {RIS}—A Unified Approach},
    journal = {IEEE Trans. on Wireless Communications},
    volume = {23},
    number={7},
    pages = {7423-7438},
    year = {2024}
}

@article{FanHoffman,
 ISSN = {00029939, 10886826},
 URL = {http://www.jstor.org/stable/2032662},
 author = {Ky Fan and A. J. Hoffman},
 journal = {Proceedings of the American Mathematical Society},
 number = {1},
 pages = {111--116},
 publisher = {American Mathematical Society},
 title = {Some Metric Inequalities in the Space of Matrices},
 urldate = {2026-01-26},
 volume = {6},
 year = {1955}
}

@article{Hassibi02,
author = {Hassibi, B. and Hochwald, B. M.},
title = {Cayley differential unitary space-time codes},
journal = {IEEE Trans. on Information Theory},
volume = {48},
number = {6},
month = {Jun.},
pages = {1485-1503},
year = {2002}
}

@article{Kovacevic06,
author = {Zhou, J. and Do, M. N.  and Kovacevic, J.},
title = {Special paraunitary matrices, Cayley transform, and multidimensional orthogonal filter banks},
journal = {IEEE Trans. on Image Processing},
volume = {15},
number = {2},
month = {Feb.},
pages = {511-519},
year = {2006}
}

@article{Zhao25,
author = {Zhao, Y. and Li, H. and Clerckx, B. and Franceschetti, M.},
title = {{MIMO} Channel Shaping and Rate Maximization Using Beyond-Diagonal {RIS}},
journal = {IEEE Trans. on Signal Processing},
volume = {73},
pages = {4397-4414},
year = {2025}
}

@book{Postikov,
author = {Postnikov, M.},
title = {Lie groups and Lie algebras},
year = {1994},
publisher = {Lectures in Geometry 5},
address = {MIR Moscow}
}

@article{Cayley46,
author = {Cayley, A.},
title = {Sur quelques propriétés des déterminants gauches},
journal = {J. Reine Angew. Math.},
volume = {32},
pages = {119-123},
year = {1846}
}

@article{Nossek24,
author = {Nossek, J. A. and  Semmler, D.  and Joham, M.  and Utschick, W. },
title = {Physically Consistent Modeling of Wireless Links With Reconfigurable Intelligent Surfaces Using Multiport Network Analysis},
journal = {IEEE Wireless Communications Letters},
volume = {13},
number = {8},
month = {Aug.},
pages = {2240-2244},
year = {2024}
}

@article{Nerini24,
author = {Nerini, M. and Shen, S. and Li, H. and Clerckx, B. },
title = {Beyond Diagonal Reconfigurable Intelligent Surfaces Utilizing Graph Theory: Modeling, Architecture Design, and Optimization},
journal = {IEEE Trans. on Wireless Commun.},
volume = {23},
number = {8},
month = {Aug.},
pages = {9972-9985},
year = {2024}
}

@article{ABRUDAN20091704,
title = {Conjugate gradient algorithm for optimization under unitary matrix constraint},
journal = {Signal Processing},
volume = {89},
number = {9},
pages = {1704-1714},
year = {2009},
issn = {0165-1684},
doi = {https://doi.org/10.1016/j.sigpro.2009.03.015},
url = {https://www.sciencedirect.com/science/article/pii/S0165168409000814},
author = {Traian Abrudan and Jan Eriksson and Visa Koivunen}
}

@article{MathisSIAM69,
 ISSN = {00361445},
 URL = {http://www.jstor.org/stable/2028117},
 author = {Robert F. Mathis},
 journal = {SIAM Review},
 number = {2},
 pages = {261--263},
 publisher = {Society for Industrial and Applied Mathematics},
 title = {Short Notes: Completion of a Symmetric Unitary Matrix},
 urldate = {2026-03-20},
 volume = {11},
 year = {1969}
}

@article{Dyson62,
 author = {Dyson, F. J.},
 journal = {J. Math. Phys.},
 number = {6},
 title = {The Threefold Way. Algebraic Structure of Symmetry Groups and Ensembles in Quantum Mechanics},
 pages = {1199--1215},
 volume = {3},
 year = {1962}
}

@book{Pozar11,
    author = {Pozar, D. M.},
    title = {Microwave Engineering, 4th Ed.},
    publisher = {John Wiley and Sons},
    year = {2011}
}

@article{PalomarFnT,
author = {Palomar, D. P. and Jiang, Y.},
title = {{MIMO} Transceiver Design via Majorization Theory},
year = {2006},
issue_date = {November 2006},
publisher = {Now Publishers Inc.},
volume = {3},
number = {4},
issn = {1567-2190},
journal = {Found. Trends Commun. Inf. Theory}
}

@article{Junior24,
 author = {Junior, W. D. S. and Guerra, D. W. M. and Filho, J. C. M. and Abrão, T.  and Hossain, E.},
 journal = {IEEE Open Journal of the Communications Society},
 title = {Manifold-Based Optimizations for {RIS}-Aided Massive {MIMO} Systems},
 pages = {7913-7940},
 volume = {5},
 year = {2024}
}

@ARTICLE{Mossallamy21,
  author={ElMossallamy, Mohamed A. and Seddik, Karim G. and Chen, Wei and Wang, Li and Li, Geoffery Ye and Han, Zhu},
  journal={IEEE Trans. on Veh. Tech.}, 
  title={{RIS} Optimization on the Complex Circle Manifold for Interference Mitigation in Interference Channels}, 
  year={2021},
  volume={70},
  number={6},
  pages={6184-6189}}

@article{Shtaiwi23,
author = {Shtaiwi, Eyad and Zhang, Hongliang and Abdelhadi, Ahmed and Swindlehurst, A. and Han, Zhu and Poor, H. Vincent},
year = {2023},
pages = {4909-4923},
title = {Sum-Rate Maximization for {RIS}-Assisted Integrated Sensing and Communication Systems With Manifold Optimization},
volume = {71},
number = [8],
journal = {IEEE Trans. on Communications}
}

@article{MarwaArxiv26,
  title={Beyond-Diagonal {RIS} For Enhanced Secrecy and Sensing Gains in Secure {ISAC} Networks: An Optimization Framework},
  author={Illi, Elmehdi and Qaraqe, Marwa},
  journal={arXiv: 10.48550/arXiv.2604.04480},
  year={2026}
}

@misc{fidanovski2026,
      title={Reciprocal Beyond-Diagonal Reconfigurable Intelligent Surface ({BD-RIS}): Scattering Matrix Design via Manifold Optimization}, 
      author={Marko Fidanovski and Iván Alexander Morales Sandoval and Hyeon Seok Rou and Giuseppe Thadeu Freitas de Abreu and Emil Björnson},
      year={2026},
      eprint={2509.20246},
      archivePrefix={arXiv},
      primaryClass={eess.SP},
      url={https://arxiv.org/abs/2509.20246}, 
}

@article{Tseng01,
author = {Tseng, P.},
year = {2001},
pages = {475-494},
title = {Convergence of a Block Coordinate Descent Method for Nondifferentiable Minimization},
volume = {109},
journal = {Journal of Optimization Theory and Applications}
}

@article{Luo13,
	Author = {M. Razaviyayn and M. Hong and Z. Q. Luo},
	Journal = {SIAM J. Opt.},
	Number = {2},
	Pages = {1126--1153},
	Title = {A Unified Convergence Analysis of Block Successive Minimization Methods for Nonsmooth Optimization},
	Volume = {23},
	Year = {2013}
}

@book{Bertsekas_nonlinear,
	author = {D. P. Bertsekas},
	edition = {2nd},
	publisher = {Athena Scientific},
	title = {Nonlinear Programming},
	year = {1999}}

@article{MaxDet,
  title={Optimal symmetric low-rank {BD-RIS} configuration maximizing the determinant of a {MIMO} link},
  author={Santamaria, I. and Soleymani, M. and Guti\'errez, J. and Jorswieck, E.},
  journal={arXiv: https://doi.org/10.48550/arXiv.2604.09335},
  year={2026}
}
\balance
\end{document}